\documentclass[11pt]{article}

\usepackage{amsmath,amssymb,amsthm,hyperref,color}
\usepackage[margin = 1in]{geometry}
\usepackage{bbm}
\usepackage{bm}
\usepackage{verbatim}
\usepackage{graphicx}
\usepackage[T1]{fontenc}
\usepackage{libertine} 

\hypersetup{colorlinks=true,citecolor=blue, linkcolor=blue, urlcolor=blue}

\newtheorem{lemma}{Lemma}[section]
\newtheorem{theorem}[lemma]{Theorem}
\newtheorem{thm}[lemma]{Theorem}

\newtheorem{fact}[lemma]{Fact}

\newtheorem{cor}[lemma]{Corollary}
\theoremstyle{remark}
\newtheorem{remark}{Remark}
\theoremstyle{definition}

\newtheorem{dfn}[lemma]{Definition}

\DeclareMathOperator*{\Var}{\mathrm{Var}}

\DeclareMathOperator*{\Exp}{\mathrm{E}}
\DeclareMathOperator*{\Ord}{\mathcal{O}}
\DeclareMathOperator*{\JC}{\Omega_\mathrm{J}}
\DeclareMathOperator*{\1}{\mathbbm{1}}

\DeclareMathOperator*{\dist}{\mathrm{dist}}
\DeclareMathOperator*{\e}{\mathrm{e}}
\DeclareMathOperator*{\boundary}{{\partial\!}}
\DeclareMathOperator*{\GD}{\mathcal{M}}
\DeclareMathOperator*{\PC}{\Omega_\mathrm{P}}

\newcommand{\PIT}{{\boldsymbol{I}_\mathcal{D}}}
\newcommand{\PI}{{\boldsymbol{I}_{\normalfont\textsc{sw}}}}
\newcommand{\PSW}{\boldsymbol{SW}}
\newcommand{\PITC}[2]{\boldsymbol{I}_{#1}^{#2}}
\newcommand{\BT}{P}

\newcommand{\df}[2]{\mathcal{E}_{#1}(#2,#2)}
\newcommand{\var}[2]{\mathrm{Var}_{#1}^{#2}}
\newcommand{\E}[2]{\mathrm{E}_{#1}^{#2}}
\newcommand{\inner}[3]{\langle #1 , #2 \rangle_{#3}}
\newcommand{\Z}{\mathbb{Z}}
\newcommand{\R}{\mathbb{R}}

\def\taumix{\tau_{\rm mix}}
\def\Tcoup{T_{\rm coup}}

\title{Spatial Mixing and Non-local Markov chains}

\author{
	Antonio Blanca\thanks{School of Computer Science, Georgia Tech, Atlanta, GA 30332.
		Email: {\tt ablanca@cc.gatech.edu}. 
		Research supported in part by NSF grants 1420934, 1563838 and 1617306.}
	\and
	Pietro Caputo\thanks{Department of Mathematics, University of Roma Tre, Largo San Murialdo 1, 00146 Roma, Italy. Email: {\tt caputo@mat.uniroma3.it}}
	\and
	Alistair Sinclair\thanks{Computer Science Division, U.C. Berkeley, Berkeley, CA 94720. Email: {\tt sinclair@cs.berkeley.edu}. Research supported in part by NSF grant 1420934.}
	\and
	Eric Vigoda\thanks{School of Computer Science, Georgia Tech, Atlanta, GA 30332.
		Email: {\tt vigoda@gatech.edu}. Research supported in part by NSF grants 1563838 and 1617306.
	\newline		
	$~~~~~~~^{\star\star}$ Part of this work was done at the Simons Institute for the Theory of Computing.}		
}

\begin{document}
	
	\maketitle	

\begin{abstract}
	
	We consider spin systems with
	nearest-neighbor interactions on 
	an $n$-vertex $d$-dimensional cube of the integer lattice graph $\Z^d$.
	We study the effects that exponential decay with distance of spin correlations,
	specifically
	the strong spatial mixing condition (SSM), 
	has on the rate of convergence to equilibrium distribution 
	of \textit{non-local} Markov chains.
	We prove that SSM implies $O(\log n)$ mixing of
	a \textit{block dynamics} whose steps can be implemented efficiently.
	We then develop 
	a methodology, 
	consisting of several new comparison inequalities concerning various block dynamics,
	that allow us to extend this result
    to other non-local dynamics.  
    As a first application of our method we prove that, if SSM holds, then the relaxation time (i.e., the inverse spectral gap) of general block dynamics is $O(r)$, where $r$ is the number of blocks. A second application of 
	our technology 
	concerns the Swendsen-Wang dynamics for the ferromagnetic Ising and Potts models.
	We show that SSM implies
	an $O(1)$ bound for the relaxation time.
	As a by-product of this implication we observe that the relaxation time of the Swendsen-Wang dynamics in square boxes of $\Z^2$ is $O(1)$
	throughout the subcritical regime of the $q$-state Potts model, for all $q \ge 2$.
	We also prove that for \textit{monotone} spin systems SSM implies that the mixing time of systematic scan dynamics is $O(\log n (\log \log n)^2)$.
	Systematic scan dynamics are widely employed in practice but have proved hard to analyze.
	Our proofs use a variety of techniques for the analysis of Markov chains including coupling, functional analysis and linear algebra.

\end{abstract}

\thispagestyle{empty}

\newpage

\setcounter{page}{1}

\section{Introduction}

Spin systems are a general framework for modeling interacting systems of simple 
elements, and arise in a wide variety of settings including statistical physics, computer
vision and machine learning (where they are often referred to as ``graphical models''
or ``Markov random fields'').  A {\it spin system\/} consists of a finite graph $G=(V,E)$
and a set~$S$ of {\it spins}; a {\it configuration\/} $\sigma\in S^V$ assigns a spin value
to each vertex $v\in V$.  For definiteness in this version of the paper, we focus on the
classical case where $G$ is a cube in the $d$-dimensional lattice $\,\Z^d$.
The probability of finding the
system in a given configuration~$\sigma$ is given by the {\it Gibbs\/} (or {\it Boltzmann\/})
distribution
\begin{equation}\label{eqn:gibbs}
\mu(\sigma) = \exp(-H(\sigma))/Z,
\end{equation}
where $Z$ is the normalizing factor (or ``partition function'') and the Hamiltonian~$H$
contains terms that depend on the spin values at each vertex (a ``vertex potential'') and
at each pair of adjacent vertices (an ``edge potential''). 
See Section~\ref{section:background} for a precise definition.

One of the most fundamental properties of spin systems is {\it (strong) 
	spatial mixing (SSM)}, which captures the fact that the correlation between spins
at different vertices decays with the distance between them (uniformly over the size of
the underlying graph~$G$)---again, see Section~\ref{section:background} for a precise definition.
SSM is closely related to the classical physical concept of a {\it phase transition}, 
which refers to the sudden disappearance of long-range correlations 
as some parameter of the system 
(typically, the edge or vertex potential) is continuously varied.\footnote{Actually
	phase transitions are usually related to a weaker notion called
	``weak spatial mixing'' (WSM); 
	in two dimensional spin systems
	WSM and SSM are known to be equivalent~\cite{MOS}.}
SSM has proved to have a number of powerful
algorithmic applications, both in the analysis of spin system dynamics (discussed in detail
below) and in the design of efficient approximation algorithms
for the partition function (a weighted generalization of approximate counting) using 
the associated self-avoiding walk trees (see, e.g., \cite{Weitz,SST,LLY,GK,SSSY,Sly,SlySun}). 

While SSM is a {\it static\/} property of a spin system, there is equal interest in 
{\it dynamic\/} properties.  By this we mean the behavior of ergodic Markov chains whose
states are the configurations of the spin system and 
whose equilibrium measure is the Gibbs distribution~(\ref{eqn:gibbs}). Such dynamics are of interest
in their own right: they provide algorithms for sampling from the Gibbs distribution
and (in many cases) are a plausible model for the evolution of the underlying system of
spins.  Of particular interest are {\it Glauber dynamics}, which at each step
pick a vertex $v\in V$ uniformly at random and update its spin in a
reversible fashion depending on the neighboring spins.   

It has been well known since pioneering work in mathematical physics from the 
late 1980s (see, e.g., \cite{Holley,AH,Zeg,SZ,MOI,MOII,Cesi}) that SSM
implies that the {\it mixing time\/} (i.e., rate of convergence) of the Glauber dynamics
is $O(|V|\log|V|)$, and hence optimal \cite{HS}; indeed, the reverse implication is also true,
so the phase transition is manifested in the mixing time of the dynamics (see, e.g., \cite{SZ,MOI,DSVW}).
The above implication was established using sophisticated functional analytic techniques,
though more recently a 
simple 
combinatorial proof was given in~\cite{DSVW} for the 
special case of {\it monotone\/} systems (where the edge potential favors pairs of equal
spins---see Section \ref{section:ss} for a precise definition). 

The intuition for these mixing time bounds comes from the fact that in the absence of long-range correlations (i.e., SSM),
the system
mimics 
the behavior of one with no interactions where 
the Gibbs distribution (\ref{eqn:gibbs}) is simply a product measure.
Consequently, \textit{local} Markov chains like the Glauber dynamics
require $\Theta(|V| \log |V|)$ steps to mix.
On the other hand, 
\textit{non-local} dynamics, where a large fraction of the configuration may be updated in a single step, could potentially converge to the Gibbs distribution much faster.
These dynamics have to contend with the possibly high computational cost of implementing a single step. However, in some cases,
non-local steps can be efficiently implemented by taking advantage 
of 
specific features of the models.

The current paper concerns 
the effects of SSM
on the rate of convergence to equilibrium of \textit{non-local} dynamics.
Our first contribution consists of
tight bounds for the mixing time and the spectral gap
of a \textit{block dynamics}.  The spectral gap
is the inverse of the \textit{relaxation time}, which measures the speed of convergence 
to the stationary distribution
when the initial configuration is reasonably close to this distribution (a ``warm start''), whereas the mixing time assumes a worst possible starting configuration. The relaxation time is another well studied notion of rate of convergence (see, e.g., \cite{JSV,KLS}). 

Let $\{A_1,\dots,A_r\}$ be a collection of sets (or blocks) such that $V = \cup_i A_i$.
A \textit{(heat-bath) block dynamics} with blocks $\{A_1,\dots,A_r\}$ is a Markov chain that in each step picks a block $A_i$ uniformly at random 
and updates the configuration in $A_i$ with a new configuration distributed according to the conditional measure in $A_i$ given the configuration in $V \setminus A_i$.
We first consider the following choice of blocks. 
Start with a regular pattern of non-overlapping $d$-dimensional lattice cubes of side $L \ll |V|^{1/d}$, with a fixed minimal distance between cubes, and let $A$ denote the union of all cubes in this pattern. By considering all possible lattice translations of the set $A\cap V$ we obtain the blocks $\{A_1,\dots,A_r\}$ where $r=O(L^d)$; see Figure~\ref{fig:tiles} on page~\pageref{fig:tiles}. Each such block $A_i$ is called a \textit{tiling} of $V$ and the associated block dynamics is called the \textit{tiled block dynamics}. We refer to Section 3 for a precise definition. 

\begin{thm}
	\label{thm:intro:parallel-block}
	When $L$ is a sufficiently large constant (independent of~$\,|V|$), SSM implies that the mixing time of the tiled block dynamics is $O(\log n)$ and that its relaxation time is $O(1)$.
\end{thm}

\noindent
In practice, the steps of the tiled block dynamics can be implemented efficiently in parallel.
However,
the main significance of this result is that,
in conjunction with a comparison methodology we develop,
it allows us to establish several new results
for standard non-local dynamics.
The first consequence of this technology is a tight bound for the relaxation time
of general block dynamics.

\begin{theorem}
	\label{thm:intro:general:block}
	SSM implies that the spectral gap 
	of any heat-bath block dynamics with $r$ blocks is $\Omega(\frac{1}{r})$,
	and hence its 
	relaxation time is $O(r)$.
\end{theorem}

\noindent
We observe that there are no restrictions on the geometry of the blocks $A_i$ in this theorem, other than $V = \cup_i A_i$. 
This optimal bound for the spectral gap was known before only for certain specific collections of blocks (see, e.g., \cite{Mart,DSVW}), 
and previous analytic methods apparently do not apply to the general setting. 

A second application of our techniques concerns the so-called
{\it Swendsen-Wang (SW)\/} dynamics~\cite{SW}. The SW dynamics is a widely studied
reversible dynamics for the ferromagnetic Ising and Potts models, which are among
the most important and classical of all spin systems. In the {\it ferromagnetic $q$-state Potts model},
there are $q$ spin values and the edge potential favors equal spins on neighbors.
More precisely, $\mu(\sigma) \propto \exp(\beta a(\sigma))$ where $a(\sigma)$
is the number of edges connecting vertices
with the same spin values in $\sigma$, and $\beta > 0$ is a parameter of the model.
The {\it Ising model\/} is just the special case
$q=2$. 

The SW dynamics is non-local, and updates the entire configuration in a single
step, according to a scheme inspired by the related random-cluster model. (The exact definition of this dynamics is given in Section \ref{section:sw}.)
We prove that the relaxation time of the SW dynamics is $\Omega(1)$, provided SSM holds.  
More formally,
let $\PSW$ be the transition matrix of the Swendsen-Wang dynamics 
for the Potts model on an $n$-vertex cube in $\Z^d$, and let $\lambda(\PSW)$ denote its spectral gap.

\begin{thm}
	\label{thm:intro:sw}
	For all $q \ge 2$, SSM implies that $\lambda(\PSW) = \Omega(1)$; hence the  
	relaxation time of the SW dynamics is $O(1)$.
\end{thm}
\noindent
This optimal bound for the spectral gap is a substantial improvement over the best previous result due to Ullrich~\cite{Ullrich1}, 
where, in $\Z^d$, SSM was shown to imply that $\lambda(\PSW) =  \Omega(n^{-1})$. 
For earlier related work
in $\Z^d$ see \cite{MOSc,CF}. 
Tight spectral gap bounds such as ours for the SW dynamics were known previously only  in the \textit{mean-field} setting, where the graph $G$ is the complete graph \cite{LNNP,GSV,BSmf}.
For other relevant work see \cite{GuoJ},
where Guo and Jerrum proved that when~$q = 2$ the SW dynamics
mixes in polynomial time on \textit{any} graph.
We note that
our spectral gap result does not immediately imply a $\mathrm{polylog}(n)$ bound on the
mixing time, as one might hope; this is because there is an inherent penalty of
$O(n)$ in relating spectral gap to mixing time, so the mixing time bound
implied by Theorem~\ref{thm:intro:sw} is $O(n)$.

In two dimensions
SSM is known to hold for all $q \ge 2$ and all $\beta < \beta_c(q)$,
where $\beta_c(q) = \log (1+\sqrt{q})$ 
is the uniqueness threshold; this is a consequence of the results in \cite{BDC,KA1,MOS}. 
Therefore, 
we have the following interesting corollary of Theorem \ref{thm:intro:sw}.

\begin{cor}
	\label{cor:sw:intro}
	In an $n$-vertex square box of $\,\Z^2$, for all $q \ge 2$ and all $\beta <\beta_c(q)$ we have
	$\lambda(\PSW) = \Omega(1)$; hence the relaxation time of the SW dynamics is $O(1)$.
\end{cor}
\noindent
In $\,\Z^2$, Ullrich's result \cite{Ullrich1} implies that the relaxation time of the SW dynamics is
$O(n)$ for $\beta < \beta_c(q)$,
$O(n^2 \log n)$ for $\beta > \beta_c(q)$, and
at most polynomial in $n$ for $\beta = \beta_c(q)$ and $q=2$.
Recently, Gheissari and Lubetzky \cite{GL,GL17},
using the
results of Duminil-Copin et al.\ \cite{DCST,DGHMT} settling the continuity of phase transition,
analyzed the dynamics at the critical point $\beta_c(q)$ for all $q$.
They showed that the mixing time is at most polynomial in $n$ for $q=3$, at most quasi-polynomial for $q=4$, and $\exp(\Omega(n))$ for $q > 4$.
Previously, Borgs et al.\ \cite{BFKTVV,BCT}
proved an exponential lower bound for the mixing time on the $d$-dimensional torus when $\beta = \beta_c(q)$, but only for sufficiently large $q$.

Our last contribution concerns 
the {\it systematic scan\/} dynamics, which is a version
of Glauber dynamics in which the vertex~$v$ to be updated is chosen not uniformly at random but
according to a fixed ordering of the vertex set~$V$; one step of systematic scan consists 
of updating each vertex $v\in V$ once according to this ordering.  Systematic scan is
widely employed in practice, and there is a folklore belief that its mixing time should be
closely related to that of standard (random update) Glauber dynamics; however, it has
proved much harder to analyze, and indeed a number of works have been devoted to
this topic (see, e.g., \cite{DGJ-1d,DGJ-dob,DGJ-norm,Hayes}). 
The best general condition under
which systematic scan dynamics is known to be rapidly mixing is due to Dyer, Goldberg
and Jerrum~\cite{DGJ-norm}, and is closely related to the Dobrushin
condition for uniqueness of the Gibbs measure; this condition in turn is known to be
stronger (and in some cases significantly stronger) than SSM \cite{SZ,MOI}.

For the special case of \textit{monotone} spin systems we can show that the systematic scan dynamics mixes in $O(\log n (\log \log n)^2)$ steps for any ordering of the vertices, whenever
SSM holds. Additionally,  
for a wide class of orderings we can show that the mixing time is $O(\log n)$, provided again that SSM holds.
For a vertex ordering~$\Ord$, let $\mathcal{L}(\Ord)$ denote the length of
the longest subsequence of $\Ord$ that is a path in $G$.  

\begin{theorem}
	\label{thm:intro:ss-monotone}
	In a monotone spin system on $\,\Z^d$,
	SSM implies that 
	the mixing time for the systematic scan dynamics
	on an $n$-vertex cube in $\,\Z^d$
	is
	$O(\log n (\log \log n)^2)$ for any ordering~$\Ord$.
	Moreover, if $\mathcal{L}(\Ord)=O(1)$ then SSM implies that the mixing time is $O(\log n)$.
\end{theorem}

\noindent
Note that the condition $\mathcal{L}(\Ord) = O(1)$ is usually easy to check in practice.
Moreover, it is easy to choose orderings~$\Ord$ for which $\mathcal{L}(\Ord)$ is 
bounded; for example, in $\,\Z^d$, $G$ is always bipartite, so the ordering~$EO$ that 
updates first all the even vertices, then all the odd ones, has $\mathcal{L}(EO) = 2$.
This particular systematic scan dynamics,
called the \textit{alternating scan dynamics},
is used in practice to sample from the Gibbs distribution and thus has received some attention \cite{PW,GKZ}.
Using our comparison technology we prove that, for \textit{general} spin systems, the relaxation time of the alternating scan dynamics is $O(1)$, provided SSM holds.

\begin{thm}
	\label{thm:intro:ss:eo}
	SSM implies that 
	the relaxation time of the alternating scan dynamics
	on an $n$-vertex cube in $\,\Z^d$
	is $O(1)$.
\end{thm}
\noindent
We emphasize that Theorem \ref{thm:intro:ss:eo} applies to general (not necessarily monotone) spin systems.
In spin systems with the SSM property, the best previously known bound 
for the relaxation time of the alternating scan dynamics was $O(n)$;
this bound follows from a recent result of Guo et al.~\cite{GKZ}. 
We observe that since the alternating scan dynamics is non-reversible, 
its relaxation time is defined in terms of the spectral gap 
of its multiplicative reversiblization; see, e.g., \cite{Fill,MT}.

The rest of the paper is organized as follows.
We conclude this introduction
with a brief discussion of our techniques. 
Section \ref{section:background} contains some basic terminology, definitions and facts used throughout the paper. 
In Section \ref{section:blocks}
we derive our results for the tiled block dynamics (Theorem \ref{thm:intro:parallel-block}) 
and introduce our comparison technology in Section \ref{subsection:comparing-tiled-block}.
In Sections \ref{section:sw} and \ref{section:partition}
we provide two applications of this technology:
bounds for the spectral gaps of the SW dynamics (Theorem \ref{thm:intro:sw}) and of the general block dynamics (Theorem \ref{thm:intro:general:block}), respectively. 
Finally, in Section \ref{section:ss} we provide our proofs for Theorems \ref{thm:intro:ss-monotone} and \ref{thm:intro:ss:eo} concerning systematic scan dynamics.

\subsection{Overview of Techniques}

We conclude this introduction by briefly indicating some of our techniques.
We use the path coupling method of Bubley and Dyer \cite{BD} to establish our results for the tiled (heat-bath) block dynamics in Theorem \ref{thm:intro:parallel-block}. 
Our proof of this theorem 
is a generalization of the methods 
in \cite{DSVW}.
We then develop a novel comparison methodology,
consisting of several new comparison inequalities concerning various block dynamics,
that together with this result allow us to establish Theorems \ref{thm:intro:general:block} and $\ref{thm:intro:sw}$. 
We provide next a high-level overview of this technology.

We consider a more general class of tiled block dynamics.
Suppose that for each $i=1,\dots,r$ and each configuration $\tau$ in $V\setminus A_i$,
we are given an ergodic Markov chain $S_i^\tau$
that acts only on the tiling $A_i$,
has $\tau$ as the fixed configuration in $V \setminus A_i$
and is reversible with respect to $\mu(\cdot|\tau)$.
Given this family of Markov chains, we consider the tiled block dynamics 
that chooses a tiling $A_i$ uniformly at random from $\{A_1,\dots,A_r\}$
and updates the configuration in $A_i$ with a step of $S_i^\tau$, provided $\tau$ is the configuration in $V \setminus A_i$.
We are able to show that the spectral gap 
of any such tiled block dynamics is determined by
the spectral gap of the tiled 
\textit{heat-bath} block dynamics (which is considered in Theorem \ref{thm:intro:parallel-block})
and the spectral gaps of the $S_i^\tau$'s.
To bound the spectral gaps of the $S_i^\tau$'s we crucially use the fact that, by design, the $A_i$'s consists of non-interacting  
$d$-dimensional cubes of constant volume.

We use this methodology 
in the proof of Theorem \ref{thm:intro:general:block}
to show that the \textit{heat-bath} block dynamics with exactly two blocks, one ``even'' block containing all the even vertices and an ``odd'' one with all the odd vertices,
has a constant spectral gap provided SSM holds.
For this, we consider
the tiled block dynamics
that picks a tiling $A_i$ uniformly at random
and with probability $1/2$ 
performs a heat-bath
update in all the even vertices in $A_i$,
and otherwise in all the odd ones.
The other part of the proof consists of
establishing a comparison inequality
between the spectral gaps of the even/odd heat-bath block dynamics (i.e., the block dynamics with exactly two blocks: the even and odd ones)
and general heat-bath block dynamics (i.e., where the collection of blocks $\{A_1,\dots,A_r\}$ is arbitrary). For this, we use two key properties of the variance functional: monotonicity and tensorization.

To derive our results for the SW dynamics in Theorem \ref{thm:intro:sw} 
we introduce an auxiliary variant of the SW dynamics
that only updates isolated vertices (instead of connected components of any size).
This isolated vertices variant 
can be compared to a tiled block dynamics that in a step updates all the isolated vertices
in a single block $A_i$ chosen uniformly at random from $\{A_1,\dots,A_r\}$. Our comparison methodology above is then used 
to show that the spectral gap of this tiled block dynamics is $\Omega(1)$.
To establish comparison inequalities between the spectral gaps of the SW dynamics, the isolated vertices variant of the SW dynamics and the tiled block dynamics
that updates isolated vertices in a tiling,
we use elementary functional analysis
and the comparison framework of Ullrich \cite{Ullrich1,Ullrich2,Ullrich4}.

The proof of our later theorem on systematic scan for monotone systems (Theorem \ref{thm:intro:ss-monotone}) is loosely based on ideas from \cite{DSVW}. Finally, to establish our result for the alternating scan dynamics (Theorem \ref{thm:intro:ss:eo}), we relate the spectral gap of this dynamics to that of
the even/odd heat-bath block dynamics, which we analyze in the proof of Theorem \ref{thm:intro:general:block}.

\section{Background}
\label{section:background}

\subsection{Spin systems}
\label{subsection:prelim:spin-systems}
Let $\mathbb{L} = (\Z^d,\mathbb{E})$ be the infinite $d$-dimensional lattice graph,
where for $u,v \in \Z^d$, $(u,v) \in \mathbb{E}$ iff ${||u-v||}_1 = 1$.
Let $V$ be a finite subset of $\,\Z^d$ and let $G = (V,E)$ be the induced subgraph.
We use $\boundary V$ to denote the boundary of $G$, i.e., the
set of vertices in $\Z^d \setminus V$ connected by an edge in $\mathbb{E}$ to $V$.

A $\textit{spin system}$ on $G$ consists of a set of {\it spins} $S =\{1,\dots,q\}$, a symmetric {\it edge potential} $U: S \times S \rightarrow \R$ and a {\it vertex potential} $W: S \rightarrow \R$.
A {\it configuration} $\sigma:V\to S$ of the system is an assignment of spins to the vertices of~$G$;
we denote by $\Omega$ the set of all configurations. A {\it boundary condition\/} $\psi$ for~$G$ is an
assignment of spins to some (or all) vertices in~$\boundary V$; i.e., $\psi: A^{\psi} \rightarrow S$ with $A^\psi \subset \boundary V$. 
The boundary condition where $A^\psi = \emptyset$ is called the \textit{free boundary condition}.

Given a boundary condition~$\psi$, each configuration $\sigma \in \Omega$ is assigned  probability
\[\mu^{\psi} (\sigma) = \frac{1}{Z} \cdot {\e}^{-H^{\psi}_G(\sigma)},\]
where $Z$ is the normalizing constant and
\[H^{\psi}_G (\sigma) = -\sum_{(u,v) \in E} U(\sigma(u),\sigma(v))\; - \sum_{(u,v) \in \mathbb{E}\,:\, u \in A^\psi, v \in V} U(\psi(u),\sigma(v))\;  - \;\sum_{u \in V} W(\sigma(u)).\]
In 
the statistical physics literature, $Z$ is
called the \textit{partition function}
and $H^{\psi}_G$ the \textit{Hamiltonian} of the system.

A particularly well known and widely studied spin system is the \textit{Ising/Potts model}, where $S\!=\{1,\!\dots\!,q\}$, $U(s_1,s_2) = \beta \cdot \1(s_1 = s_2)$ and $W(s) = \beta h_s$.
The parameter $\beta \in \R$ is related to the inverse temperature of the system and $(h_1,...,h_q) \in {\R}^q$ to an external magnetic field.
In Section \ref{section:sw} we analyze dynamics for the Ising/Potts model with ferromagnetic interactions ($\beta > 0$) and no external field ($h_i = 0$ for all $i$).

\begin{remark}
	There are important spin systems, such as the hard-core model and the antiferromagnetic Potts model at zero temperature (proper $q$-colorings), that require the edge potential $U$ to be infinite for certain configurations; namely, there are \textit{hard constraints} in the system that make certain configurations
	invalid. 
	Our results in Sections \ref{section:blocks}, \ref{section:partition} and \ref{section:ss}
	hold in this more general setting provided the system is \textit{permissive}. 
	A spin system is permissive if for any $V\subset \Z^d$
	and any configuration $\tau$ on $\Z^d \setminus V$, there is at least one configuration $\sigma$ on $V$ such that 
	$\mu(\sigma|\tau) > 0$.
	This ensures that the measure $\mu(\cdot|\tau)$ is well-defined.
	It is easy to verify that, in addition to systems without hard constraints,
	the hard-core model for all $\lambda > 0$ 
	and proper $q$-colorings when $q \ge 2d+1$ are all permissive systems.
\end{remark}

\subsection{Glauber dynamics}
\label{subsection:prelim:markov-chains}
Consider the spin system $(S=\{1,\dots,q\},U,W)$ on $G=(V,E)$ with a fixed boundary condition $\psi$.
Let $\GD$ be a Markov chain that, given a configuration $\sigma$ on $V$, performs the following update:
\begin{enumerate}
	\item Pick $v \in V$ uniformly at random (u.a.r.);
	\item Replace $\sigma(v)$ with a spin from $S = \{1,...,q\}$ sampled according to the distribution $\mu(\cdot | \sigma(V\setminus v))$.
\end{enumerate}
This Markov chain is called the \textit{(heat-bath) Glauber dynamics}.
$\GD$  is clearly reversible with respect to (w.r.t.) $\mu^{\psi}$ and, to avoid complications, we assume that it is irreducible. 
(This is always the case in systems without hard constraints, but  $\GD$ could be reducible for some permissive systems; e.g., proper $q$-colorings when $q = 2d+1$.)

\subsection{Strong spatial mixing (SSM)}
\label{subsection:ssm}
Several notions of decay of correlations in spin systems have been useful in the analysis of local algorithms. A particularly important one is
SSM, which says that the influence of a set on another decays exponentially with the distance between these sets.

For a fixed finite $V \subset \Z^d$ and $a,b > 0$, let $\mathcal{C}(V,a,b)$ be the condition that 
for all $B \subset V$, all $u \in \boundary V$, and any pair of boundary conditions $\psi$, $\psi_u$ on $\boundary V$ that differ only at $u$, we have
\begin{equation}
\label{prelim:eq:ssm}
{\|\mu_B^\psi \,-\, \mu_B^{\psi_u}\|}_{\textsc{tv}} \,\,\le\,\,  b\, \exp(-a \cdot \dist(u, B)),
\end{equation}	
where 
$\mu_B^\psi$ and $\mu_B^{\psi_u}$ are the probability measures induced in $B$ by $\mu^\psi$ and $\mu^{\psi_{u}}$, respectively,
$\|\cdot\|_{\textsc{tv}}$ denotes total variation distance
and $\dist(u, B) = \min_{v \in B} {{\|u-v\|}_1}$.		
\begin{dfn}	
	\label{def:prelim:ssm}
	A spin system on $\,\Z^d$ has SSM if
	there exist $a,b > 0$ such that
	$\mathcal{C}(\Lambda,a,b)$ holds for every
	$d$-dimensional cube $\Lambda \subset \Z^d$.		
\end{dfn}

\begin{remark}
	The definition of SSM varies in the literature.
	The main difference lies in the class of subsets $V \subset \Z^d$ for which $\mathcal{C}(V,a,b)$ is required to hold.
	The two boundary conditions may also differ on a larger subset of $\boundary V$.
	We work here with one of the weakest versions of SSM.
	In particular, this notion is known to hold
	for the Ising/Potts model on $\Z^2$ for all $q \ge 2$ and $\beta < \beta_c(q)$, where $\beta_c(q)$ is the uniqueness threshold.
\end{remark}

\subsection{Mixing and coupling times}
\label{subsection:coupling}
Let $M$ be an ergodic Markov chain over $\Omega$ with stationary distribution $\mu^\psi$.
Let $M^t(X_0,\cdot)$ denote the distribution of $M$ after $t$ steps starting from $X_0 \in \Omega$, and let
\[
\taumix(M,\varepsilon) = \max\limits_{X_0 \in \Omega}\min \left\{ t \ge 0 : {\|{M}^t(X_0,\cdot)-\mu^\psi\|}_{\textsc{tv}} \le \varepsilon \right\}.
\]
The \textit{mixing time} of $M$ is defined as $\taumix(M) = \taumix(M,1/4)$. 

A {\it (one step) coupling} of the Markov chain $M$ specifies, for every pair of states $(X_t, Y_t) \in \Omega\times\Omega$, a probability distribution over $(X_{t+1}, Y_{t+1})$ such that the processes $\{X_t\}$ and $\{Y_t\}$, viewed in isolation, are faithful copies of $M$, and if $X_t=Y_t$ then $X_{t+1}=Y_{t+1}$. 
Let $\Tcoup(\varepsilon)$
be the minimum $T$ such that $\Pr[X_T \neq Y_T] \le \varepsilon$, maximized over pairs of initial configurations $X_0$, $Y_0$. 
The following inequality is standard:
$$\taumix(M,\varepsilon) \le \Tcoup(\varepsilon);$$
(see, e.g., \cite{LPW}).
The {\it coupling time} is $\Tcoup = \Tcoup(1/4)$ and thus 
$\taumix(M) \le \Tcoup$.
Moreover, if $T = k \cdot \Tcoup$ for any positive integer $k$, then
\begin{equation}
\label{eq:prelim:coup-boost}
\Pr[X_T \neq Y_T] \le 1/4^k.
\end{equation}

\subsection{Analytic tools}
\label{subsection:prelim:analytic-tools}
Our proofs 
use elementary notions from functional analysis, which we briefly review here. For extensive background on the application of such ideas to the analysis of finite Markov chains, see~\cite{SC,MT}. 

Let $P$ be the transition matrix of a finite irreducible Markov chain with state space $\Omega$ and stationary distribution $\mu$. 
For any $f\in\R^{|\Omega|}$, we let $Pf(x)=\sum_{y\in\Omega}P(x,y)f(y)$.
If we endow $\R^{|\Omega|}$ with the inner product 
$\langle f,g \rangle_\mu = \sum_{x \in \Omega} f(x)g(x)\mu(x)$,
we obtain a Hilbert space denoted $L_2(\mu) = (\R^{|\Omega|},\langle \cdot,\cdot \rangle_\mu)$ and $P$ defines an operator from $L_2(\mu)$ to $L_2(\mu)$. The Cauchy-Schwarz inequality implies 
\begin{equation}
\label{eq:sw:contraction}
\inner{f}{Pf}{\mu} \le  \inner{f}{f}{\mu}.
\end{equation}
Consider two Hilbert spaces $S_1$ and $S_2$ with inner products $\langle\cdot,\cdot\rangle_{S_1}$ and $\langle\cdot,\cdot\rangle_{S_2}$ respectively, and let $K:S_2 \rightarrow S_1$ be a bounded linear operator. The \textit{adjoint} of $K$ is the unique operator $K^*:S_1 \rightarrow S_2$ satisfying $\langle f,Kg \rangle_{S_1} = \langle K^*f,g \rangle_{S_2}$ for all $f \in S_1$ and $g \in S_2$. If $S_1 = S_2$, $K$ is \textit{self-adjoint} when $K = K^*$.

In our setting, the adjoint of $P$ in $L_2(\mu)$ is given by the transition matrix $P^*(x,y)=\mu(y)P(y,x)/\mu(x)$, and therefore $P$ is self-adjoint iff $P$ is reversible w.r.t.\ $\mu$. In this case the spectrum of $P$ is real and we let $1 = \lambda_1 > \lambda_2 \ge ... \ge \lambda_{|\Omega|} \geq -1$ denote its eigenvalues ($1 > \lambda_2$ because $P$ is irreducible).  
The {\it absolute spectral gap} of $P$ is defined by $\lambda(P) = 1 - \lambda^*$, where $\lambda^* = \max\{|\lambda_2|,|\lambda_{|\Omega|}|\}$. If $P$ is ergodic (i.e., irreducible and aperiodic), then $\lambda(P)>0$, and 
it is a standard fact that for all $\varepsilon>0$ all reversible Markov chains satisfy 
\begin{equation}
\label{prelim:eq:gap-lower}
\taumix(P,\varepsilon) \ge \left(\lambda(P)^{-1}-1\right) \log \left(
\frac{1}{2\varepsilon}\right),
\end{equation}
(see Theorem 12.4 in \cite{LPW}). 
$\lambda^{-1}(P)$ is called 
the \textit{relaxation time}.

$P$ is \textit{positive semidefinite} if $P=P^*$ and $\langle  f,P f \rangle_{\mu} \ge 0$, $\forall f \in \R^{|\Omega|}$. In this case $P$ has only nonnegative eigenvalues. 
The \textit{Dirichlet form} of a reversible Markov chain is defined as
$$\mathcal{E}_P(f,f) = \inner{f}{(I-P)f}{\mu} = \frac{1}{2} \sum_{x,y \in \Omega} \mu(x) P(x,y) (f(x) - f(y))^2,$$		
for any $f \in \R^{|\Omega|}$. If $P$ is positive semidefinite, then
the absolute spectral gap of $P$ satisfies
\begin{equation}
\label{eq:sw:gap}
\lambda(P) = 1 - \lambda_2 = \min_{f \in \R^{|\Omega|},  \Var\nolimits_{\mu} (f) \neq 0} \frac{\mathcal{E}_{P} (f,f)}{\Var_{\mu} (f)},
\end{equation}	
where $\Var_{\mu}(f) = \sum_{x \in \Omega} (f(x)-\mu(f))^2 \mu(x)$ 
and $\mu(f) = \sum_{x \in \Omega} f(x)\mu(x)$.

\section{SSM and tiled block dynamics for general spin systems}
\label{section:blocks}

Let $V \subset \Z^d$ be a $d$-dimensional cube of volume\footnote{For $A \subset \Z^d$, the volume of $A$ is $|A|$.} $n$. Let $G=(V,E)$ be the induced subgraph
and let $\psi$ be a fixed boundary condition on $\boundary V$. For ease of notation we set $\mu = \mu^\psi$.

Let $\{A_1,\dots,A_r\}$ be a collection of sets (or blocks) such that $V = \cup_i A_i$.
A \textit{block dynamics} w.r.t.\ this collection of sets is a Markov chain that in each step picks a set $A_i$ uniformly at random from $\{A_1,\dots,A_r\}$ and updates the configuration in $A_i$.
The \textit{heat-bath block dynamics} corresponds to the case
where the configuration in $A_i$ is replaced by a new configuration distributed according to the conditional measure in $A_i$ given the configuration in $V \setminus A_i$.

In this section we consider two different versions of the block dynamics for a particular collection of sets, that with slight abuse of terminology we call \textit{tilings}.
The steps of this dynamics can be efficiently implemented in parallel, so
we believe it is interesting in its own right.
Moreover, the mixing time and spectral gap bounds we derive here
will be crucially used 
later in our proofs in Sections \ref{section:sw} and \ref{section:partition}, where we consider the SW dynamics and general block dynamics, respectively.

We define the collection of blocks first, which we denote $\mathcal{D}$.
Let $L \ll n^{1/d}$ be an odd integer.
For each $x_i \in \{0,\dots,L+2\}^d \subset \Z^d$,
let $C(x_i)$ be the union of all $d$-dimensional cubes of side length $L-1$ with centers at $x_i+\vec{h}(L+3)$ for some $\vec{h} \in \Z^d$.
The cubes in $C(x_i)$ 
have volume $L^d$ and are at distance $4$ from each other (see Figure \ref{fig:tiles}).
For each $x_i \in \{0,\dots,L+2\}^d$,
let $B_i = C(x_i) \cap V$ and let $\mathcal{D} = \{B_1,B_2,\dots,B_m \}$;  then $m=(L+3)^d$.  
We call each $B_i$ a \textit{tiling}  of $V$ since it corresponds to a tiling of $\Z^d$ with cubes of side lengh $L+3$. 
Any block dynamics w.r.t.\ $\mathcal{D}$ is called
a \textit{tiled block dynamics}.

\begin{remark}
	In our proofs we will choose $L$ to be a sufficiently large constant independent of $n$. 
	The choice of the distance $4$
	between the $d$-dimensional cubes is so that neighboring cubes do not interact.
	This distance is sufficient because we are considering spin systems 
	with only nearest-neighbor interactions. To extend our proofs to arbitrary finite range spin systems on $\Z^d$
	it suffices to choose a larger distance between these cubes.
\end{remark}	

	\begin{figure*}
	\begin{center}
		\includegraphics[page=1,clip,trim=70 590 0 72]{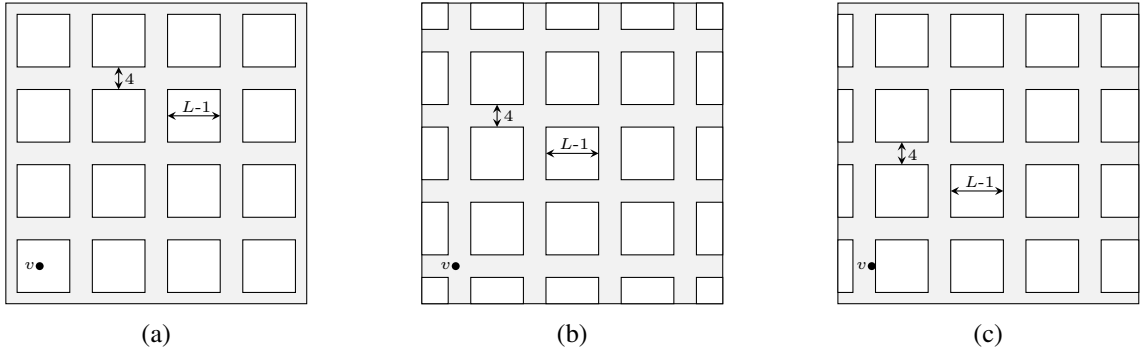}
		\caption{Three distinct tilings of $V$. In (a) vertex $v$ is in the interior of the tiling; in (b) vertex $v$ is in the exterior; and in (c) vertex $v$ is right on the boundary of the tiling. }     
		\label{fig:tiles}
	\end{center}
\end{figure*}

Let $\mathcal{B}_\mathcal{D}$ be the transition matrix of the \textit{heat-bath tiled block dynamics}. That is, given
a configuration $\sigma_t \in \Omega$ at time $t$, the chain proceeds as follows:
\begin{enumerate}
	\item Pick $k \in \{1,...,m\}$ u.a.r.;
	
	\item Update the configuration in $B_k$ with a sample from $\mu(\cdot \mid \sigma_t(V \setminus B_k))$.
\end{enumerate}	
This chain is clearly ergodic and reversible w.r.t.\ $\mu$. 
We prove the following lemma, which corresponds to Theorem \ref{thm:intro:parallel-block} from the introduction.

\begin{lemma}
	\label{lemma:heat-bath-parallel-block}
	When $L$ is a sufficiently large constant (independent of $n$), SSM implies that $\taumix(\mathcal{B}_\mathcal{D})=O(\log n)$ and $\lambda(\mathcal{B}_\mathcal{D})  \ge 1/8$.
\end{lemma}

\begin{proof}
	The proof is a generalization
	of the path coupling argument in \cite{DSVW}.	
	Let $X_t$ and $Y_t$ be two copies of the tiled heat-bath block dynamics $\mathcal{B}_\mathcal{D}$ that differ at a single vertex $v \in V$.
	We construct a coupling of the steps of $\mathcal{B}_\mathcal{D}$
	such that the expected number of disagreements between $X_{t+1}$ and $Y_{t+1}$ is strictly less than one.
	
	The region chosen in step 1 of the chain is the same in both copies.
	For every tiling $B_k$ there are three possibilities (see Figure \ref{fig:tiles}): 
	\begin{enumerate}
		\item[(a)] $v \in B_k$, in which case we use the same configuration for $B_k$ in both copies and so $X_{t+1}=Y_{t+1}$ with probability 1;
		\item[(b)] $v \in V \setminus (B_k \cup \boundary B_k)$, and again we use the same configuration to update $B_k$ in both copies. Then, $X_{t+1}$ and $Y_{t+1}$ differ only at $v$ with probability 1; or
		\item[(c)] $v \in \boundary B_k$. In this case disagreements could propagate from $v$ to the interior of $B_k$, but we describe next a coupling that limits the extent of such propagation.
	\end{enumerate}
	
	\noindent
	Case (a) occurs with probability
	${L^d}/{(L+3)^d} \ge 1/2,$
	for large enough $L$.
	Let us consider case (c); i.e., $v \in \boundary B_k$.
	This case occurs with probability at most 
	${2dL^{d-1}}/{(L+3)^d} \le {2d}/{L}$.
	Moreover, $v$ is in the boundary of exactly one of the smaller cubes (of side length at most $L-1$) in $B_k$, which we denote $\Lambda$. 
	The cube $\Lambda$ can be partitioned into the sets of vertices that are close and far from $v$.
	More precisely, let
	$R = \frac{1}{2}\left(\frac{L}{8d}\right)^{1/d}$, 
	$C = \{u \in \Lambda: \dist(u,v) \le R\}$
	and $F = \Lambda \setminus C$. SSM implies
	\[\|\mu_{F}^{\psi} \,\,-\,\, \mu_{F}^{\psi_v}   \|_{\textsc{tv}} \le b \exp\{-a\dist(v,F)\},\]
	where $\psi$ and $\psi_v$ are the two boundary conditions induced in $\Lambda$ by $X_t$ and $Y_t$, respectively, and thus differ only at $v$. 
	This implies that there is a coupling of the distributions $\mu_{F}^{\psi}$ and $\mu_{F}^{\psi_v}$ such that if $(Z_1,Z_2)$ is a sample from this coupling (so, $Z_1$ and $Z_2$ are configurations on $F$), then
	\[\Pr[Z_1 \neq Z_2]  \le b \exp\{-a\dist(v,F)\} \le b \exp\{-aR\} \le \frac{1}{L^d},\]
	where the last inequality holds for large enough $L$. 
	Hence, we can couple the update on $\Lambda$ 
	such that $X_{t+1}$ and $Y_{t+1}$ disagree on $F$ with probability at most $L^{-d}$.
	Then, the expected number of disagreements in $\Lambda$ is crudely bounded by 
	$$|C| + \frac{|F|}{L^{d}} \le (2R)^d + 1 \le \frac{L}{8d}+1.$$
	The same configuration is used to update both copies in $B_k \setminus \Lambda$ and so $X_{t+1}(B_k \setminus \Lambda) = Y_{t+1}(B_k \setminus \Lambda)$ with probability one. This is possible because the configuration in the boundary of $B_k \setminus \Lambda$ is the same in both $X_t$ and $Y_t$.
	
	Combining all these facts, we get there is a coupling such that the expected number of disagreements at time $t+1$ is at most:
	\[1 - \frac{1}{2} + \frac{2d}{L}\left( \frac{L}{8d} + 1\right) = \frac{3}{4} + \frac{2d}{L} \le \frac{7}{8},\]
	provided that $L$ is large enough. 
	The path coupling method \cite{BD} then implies that
	\[\max_{\sigma \in \Omega} \|\mathcal{B}_\mathcal{D}^t(\sigma,\cdot) - \mu(\cdot) \|_{\textsc{tv}} \le n \left(\frac{7}{8}\right)^t.\]
	This implies that 
	the mixing time of $\mathcal{B}_\mathcal{D}$ is $O(\log n)$ and that
	$\lambda^*(\mathcal{B}_\mathcal{D}) \le 7/8$ (see, e.g., Corollary 12.6 in \cite{LPW}); hence,  $\lambda(\mathcal{B}_\mathcal{D}) = 1 - \lambda^*(\mathcal{B}_\mathcal{D}) \ge 1/8$ as claimed.
\end{proof}

\subsection{Comparing tiled block dynamics}
\label{subsection:comparing-tiled-block}

In this subsection
we introduce a more general class of tiled block dynamics
and relate the spectral gaps
of the dynamics in this class 
to that of the heat-bath tiled block dynamics.
This will allow us to deduce bounds 
for the spectral gaps of various tiled block dynamics, a key step in our comparison methodology.

Each dynamics in this class
chooses a tiling $B_k$ uniformly at random from $\mathcal{D}$
and updates the configuration in $B_k$ in a reversible fashion.
Formally,
for each $1 \le k \le m$ and each valid configuration $\tau$ in $B_k^c = V \setminus B_k$, 
let $S_k^\tau$ be the transition matrix of an ergodic Markov chain 
whose state space
is the set of valid configurations in $B_k$ given that $\tau$ is the configuration in $B_k^c$.
That is,  $S_k^\tau$ is a Markov chain acting on the specific tiling $B_k$
with $\tau$ as the \textit{fixed} configuration in the exterior of $B_k$.
We assume that, for each $k$ and $\tau$, $S_k^\tau$ is 
reversible w.r.t.\ $\mu(\cdot | \tau)$ and positive semidefinite.
Using the $S_k^\tau$'s we define a tiled block dynamics as follows.
Given
a spin configuration $\sigma_t \in \Omega$, consider the chain that performs the following update to obtain $\sigma_{t+1} \in \Omega$:	
\begin{enumerate}
	\item Pick $k \in \{1,...,m\}$ u.a.r.;
	
	\item If $\tau = \sigma_t(B_k^c)$, let $\sigma_{t+1}(B_k^c) = \tau$ and perform a step of $S_k^\tau$ to obtain $\sigma_{t+1}(B_k)$. 
\end{enumerate}	
Let $S_\mathcal{D}$ denote the transition matrix of this chain. The ergodicity and reversibility of $S_\mathcal{D}$ w.r.t.\ $\mu$ follow from the ergodicity and reversibility of the $S_k^\tau$'s w.r.t.\ $\mu(\cdot | \tau)$.
We establish the following inequality between the spectral gaps of $\mathcal{B}_\mathcal{D}$ and $S_\mathcal{D}$. For $A \subset V$, let $\Omega(A)$ be the set of the valid configurations of $A$. Then,

\begin{lemma}
	\label{lemma:parallel-block}
	$$\lambda(S_\mathcal{D}) \ge \lambda(\mathcal{B}_\mathcal{D}) \min\limits_{k=1,\dots,m} \min\limits_{\tau \in \Omega(B_k^c)} \lambda(S_k^\tau).$$
\end{lemma}

\noindent
In words, this inequality states
that the spectral gap of a generic tiled block dynamics $S_\mathcal{D}$
is bounded from below by the spectral gap of the tiled \textit{heat-bath} block dynamics
times the smallest spectral gap of any of the $S_k^\tau$'s.
This is indeed a natural inequality since
roughly $\lambda^{-1}(S_k^\tau)$ steps of $S_k^\tau$ should be enough to simulate
one step of $\mathcal{B}_\mathcal{D}$ in $B_k$ when $\tau$ is the configuration in $B_k^c$.
Lemmas \ref{lemma:heat-bath-parallel-block} and \ref{lemma:parallel-block} put together
allow us to bound the spectral gap of a general class of tiled block dynamics,
provided that SSM holds and that
we know the spectral gaps of the $S_k^\tau$'s.
As we shall see in our later applications of these results, 
the geometry of the tilings in $\mathcal{D}$ 
was chosen in a way that facilitates 
the analysis of many natural choices of the $S_k^\tau$'s.


Before proving Lemma \ref{lemma:parallel-block} we state the two
standard properties of heat-bath updates which will be used in the proof. 
For $A \subset V$ let $K_A$ be the transition matrix that corresponds to a heat-bath update in the set $A$. That is,
for $\sigma,\sigma' \in \Omega$, 
$$
K_A(\sigma,\sigma') = \1(\sigma(A^c) = \sigma'(A^c)) \mu(\sigma'(A) \mid \sigma(A^c)).
$$
For ease of notation let $\mathcal{E}_{A}$ denote the Dirichlet form of $K_A$; i.e., $\mathcal{E}_{A} = \mathcal{E}_{K_A}$.  
\begin{fact} 
	\label{fact:var-heat-bath-positive}
	\label{fact:var-heat-bath}
	$K_A$ is positive semidefinite. Moreover, for any $f \in \R^{|\Omega|}$
	\[\df{A}{f} = \sum_{\tau \in \Omega(A^c)} \var{A}{\tau}(f) \mu(\tau),\]
	where 
	$\var{A}{\tau} (f) = \E{A}{\tau} [(f - \E{A}{\tau}[f])^2]$ and
	$\E{A}{\tau} [f] = \sum_{\sigma \in \Omega(A)} f(\sigma \cup \tau)\mu(\sigma \mid \tau)$. 
\end{fact}
\noindent
We proceed with the proof of Lemma \ref{lemma:parallel-block}.

\begin{proof}[Proof of Lemma \ref{lemma:parallel-block}]
	
		Let 
		$f \in \R^{|\Omega|}$. Since $\mathcal{B}_\mathcal{D} = \frac{1}{m} \sum_{k=1}^m K_{B_k}$,
		\begin{equation}
		\label{eq:df:H}
		\df{\mathcal{B}_\mathcal{D}}{f} = \frac{1}{m} \sum_{k=1}^m \df{B_k}{f}  = \frac{1}{m} \sum_{k=1}^m \sum_{\tau \in \Omega(B_k^c)}  \var{B_k}{\tau} (f) \,\, \mu (\tau),
		\end{equation}
		by Fact \ref{fact:var-heat-bath}.
		
	   	For $\tau \in \Omega(B_k^c)$, 
	   	let $\Omega^\tau(B_k)$ be the set of valid configurations on $B_k$ given that $\tau$ is the configuration on $V \setminus B_k$.
	   	For $f \in \R^{|\Omega|}$, let $f_{\tau} \in \R^{|\Omega^\tau(B_k)|}$ be 
	   	such that $f_{\tau} (\sigma) = f(\sigma \cup \tau)$ for any $\sigma \in \Omega^\tau(B_k)$.  
	   	By assumption,  $S_k^\tau$ is positive semidefinite, ergodic and reversible w.r.t.\ $\mu(\cdot\mid\tau)$. Since also
	   	$\var{\mu(\cdot| \tau)}{} (f_{\tau}) = \var{B_k}{\tau} (f),$
	   	from (\ref{eq:sw:gap}), we get
	   	\begin{equation}
	   	\label{eq:local-gap:var}
	   	0 < \lambda(S_k^\tau) \le \frac{\df{S_k^\tau}{f_{\tau}}}{\var{\mu(\cdot| \tau)}{} (f_{\tau}) } = \frac{\df{S_k^\tau}{f_{\tau}}}{\var{B_k}{\tau} (f)}.
	   	\end{equation}
	   	Let
	   	$$
	   	\lambda_{\textrm{min}} = \min_{k=1,\dots,m} \min_{\tau \in \Omega(B_k^c)} \lambda(S_k^\tau).
	   	$$
	   	Then, from the definition of the Dirichlet form, (\ref{eq:df:H}) and (\ref{eq:local-gap:var}) we get
	   	\begin{align}
	   	\label{eq:sd-skt}
	   	\df{S_\mathcal{D}}{f} 
	   	&= \frac1m \sum_{k=1}^m \sum_{\tau \in \Omega(B_k^c)} \mu(\tau)\df{S_k^\tau}{f_\tau} \\
	   	&	\geq \frac1m \sum_{k=1}^m \sum_{\tau \in \Omega(B_k^c)} \mu(\tau)\lambda(S_k^\tau)\var{B_k}{\tau} (f) 
	    \geq \lambda_{\textrm{min}}\df{\mathcal{B}_\mathcal{D}}{f} \notag.
	   	\end{align}
	   	
	   	Finally, we claim that both $\mathcal{B}_\mathcal{D}$ and $S_\mathcal{D}$ are positive semidefinite.
	   	$\mathcal{B}_\mathcal{D}$ is an average over heat-bath updates each of which is positive semidefinite by Fact \ref{fact:var-heat-bath-positive}. 
	   	Hence, 
	   	$\mathcal{B}_\mathcal{D}$ is positive semidefinite. 
	   	Similarly, the positivity of $S_\mathcal{D}$ follows from the fact that by assumption
	   	the $S_k^{\tau}$'s are positive semidefinite.
	   	Indeed, from (\ref{eq:sd-skt}) and the definition of Dirichlet form, we get
	   	\begin{align*}
	   	\inner{f}{S_\mathcal{D} f}{\mu}
	   	= \frac1m \sum_{k=1}^m \sum_{\tau \in \Omega(B_k^c)} \mu(\tau)\inner{f_\tau} {S_k^{\tau} f_\tau}{\mu(\cdot|\tau)} \ge 0.
	   	\end{align*}
		Therefore, by (\ref{eq:sw:gap}), $\lambda(S_\mathcal{D})  \ge \lambda(\mathcal{B}_\mathcal{D}) \lambda_{\rm min}$, as claimed.
\end{proof}

We conclude this section with the proof of Fact \ref{fact:var-heat-bath}.

\begin{proof}[Proof of Fact \ref{fact:var-heat-bath}]
	Since $K_A = K_A^* = K_A^2$, $K_A$ positive semidefinite.
	For $\tau \in \Omega(A^c)$, 
	let $\Omega^\tau(A)$ be the set of valid configurations on $A$ when the configuration on $V \setminus A$ is $\tau$. Then,
	by the definition of the Dirichlet form,
	\begin{align}
	\df{A}{f} 
	&= \frac{1}{2} \sum_{\tau \in \Omega(A^c)} \sum_{\sigma,\sigma' \in \Omega^\tau(A)} \mu(\sigma \cup \tau) \mu(\sigma' \mid \tau) (f(\sigma\cup\tau) - f(\sigma'\cup\tau))^2\notag\\
	&= \frac{1}{2} \sum_{\tau \in \Omega(A^c)} \mu(\tau) \sum_{\sigma,\sigma' \in \Omega^\tau(A)} \mu(\sigma \mid \tau) \mu(\sigma' \mid \tau) (f(\sigma \cup \tau) - f(\sigma'\cup\tau))^2 \notag\\
	&= \sum_{\tau \in \Omega(A^c)} \var{A}{\tau}(f) \mu(\tau) \notag.\qedhere
	\end{align}
\end{proof}

\section{SSM and the Swendsen-Wang dynamics for the Potts model}
\label{section:sw}

In this section we 
show that SSM implies fast mixing of the {\it Swendsen-Wang (SW)\/} dynamics.
In particular, we prove that when $V \subset \Z^d$ is a finite $d$-dimensional cube,
the \textit{relaxation time} (i.e., the inverse spectral gap)
of the SW dynamics on the graph induced by $V$ is at most $O(1)$, provided the system has SSM.

The SW dynamics is a {\it non-local\/} Markov chain
for the ferromagnetic Potts model ($\beta > 0$) with no external field ($h_i = 0$ for all $i$); see Section \ref{subsection:prelim:spin-systems} for the definition of this model.
The state space of the SW dynamics is the set
of Potts configurations $\PC$, and it is straightforward to verify the reversibility of this chain w.r.t.\ the Potts measure, which, for distinctness, we will denote $\pi$ (see, e.g., \cite{ES}).
We focus here on the free boundary condition case for clarity, but our results hold without significant modifications for the SW dynamics with arbitrary boundary conditions.

Let $V \subset \Z^d$ be a  $d$-dimensional cube of volume $n$ and let $G=(V,E)$ be the induced subgraph.
Given a Potts configuration $\sigma_t$, a step of the SW dynamics
results in a new configuration $\sigma_{t+1}$ as follows:
\begin{enumerate}
\item Add each \textit{monochromatic} edge independently with probability $p= 1 - e^{-\beta}$
to obtain a \textit{joint} configuration~$(A_t,\sigma_t)$, where $A_t \subseteq E$ and an edge $(u,v)$ is monochromatic if $\sigma_t(u)=\sigma_t(v)$;

\item Assign to each connected component of $(V,A_t)$ independently a new spin from $\{1,\dots,q\}$ u.a.r.;

\item Remove all edges to obtain the new Potts configuration $\sigma_{t+1}$.
\end{enumerate}	

\noindent
Let $\PSW$ be the transition matrix of the SW dynamics on $G$. 
In this section we prove Theorem \ref{thm:intro:sw} from the introduction. 
Corollary \ref{cor:sw:intro}
follows directly from Theorem \ref{thm:intro:sw} and the fact that, in $\Z^2$, SSM holds for all $\beta < \beta_c(q)$ and $q \ge 2$ (see \cite{BDC,KA1,MOS}).
In the proof of Theorem \ref{thm:intro:sw} we use several auxiliary Markov chains that we define and briefly motivate in 
Section \ref{subsection:mc}. The proof of Theorem \ref{thm:intro:sw} is then provided in Section~\ref{subsection:proofs}.

\subsection{Auxiliary Markov chains}
\label{subsection:mc}

In Section \ref{section:blocks} we established that the spectral gap of the heat-bath tiled block dynamics is at least $1/8$, provided SSM holds (see Lemma \ref{lemma:heat-bath-parallel-block}). To prove Theorem \ref{thm:intro:sw}
we show that the spectral gap of the SW dynamics is at least the spectral gap of the heat-bath tiled blocked dynamics times a constant that depends only on $\beta$, $L$ and $d$.
Establishing such inequality directly
seems difficult because the SW dynamics could change the spins in a large component intersecting many of the $d$-dimensional cubes in a tiling. To work around this issue we introduce the following Markov chain.

\bigskip\noindent
\textbf{Isolated vertices (SW) dynamics} ${\PI}$\textbf{. }
Consider the Markov chain that, given a Potts configuration $\sigma_t$ at time $t$, performs the following update to obtain $\sigma_{t+1}$:	
\begin{enumerate}
	\item Add each monochromatic edge independently with probability $p$ to obtain $(A_t \subseteq E,\sigma_t)$;		
	\item Assign to each \textit{isolated} vertex of $(V,A_t)$ independently a new spin from $\{1,\dots,q\}$ u.a.r.; 
	\item Remove all edges to obtain $\sigma_{t+1}$.
\end{enumerate}	
We call this chain the \textit{isolated vertices dynamics} and with a slight abuse of notation we let $\PI$ also denote its transition matrix. Intuitively, the SW dynamics ought to be faster than the isolated vertices dynamics since it updates all the components
of any size simultaneously, instead of just the isolated vertices. We show that this is indeed the case.
\begin{lemma}
	\label{lemma:sw:gap-ineq-1}
	$\lambda(\PSW) \ge \lambda(\PI)$.
\end{lemma}
\noindent
The proof of this lemma is given in Section \ref{subsection:sw:comparison}.
The motivation for introducing ${\PI}$ is that now we can easily define a tiled variant of this chain as follows. 

\bigskip\noindent
\textbf{Isolated vertices tiled dynamics} $\PIT$\textbf{. }
Recall that $\mathcal{D} = \{B_1,\dots,B_m\}$ is the collection of tilings; see Section~\ref{section:blocks} for the precise definition. 
Given a Potts configuration $\sigma_t$, one step of the \textit{isolated vertices tiled dynamics} is given by:
\begin{enumerate}
	\item Add each monochromatic edge independently with probability $p$ to obtain $(A_t \subseteq E,\sigma_t)$;	
	\item Pick $k \in \{1,...,m\}$ u.a.r.;	
	\item Assign to each \textit{isolated} vertex in $B_k$ independently a new spin from $\{1,\dots,q\}$ u.a.r.; 
	\item Remove all edges to obtain $\sigma_{t+1}$.
\end{enumerate}
We use $\PIT$ to denote the transition matrix of this chain. 
Intuitively, $\PI$ should reach equilibrium faster than $\PIT$ 
since in each step it updates the spins of all isolated vertices, instead of just those in a single tiling. This intuition is made rigorous in the following lemma, which is proved in Section \ref{subsection:sw:comparison}.
\begin{lemma}
	\label{lemma:sw:gap-ineq-2}
	$\lambda(\PI) \ge \lambda(\PIT)$.
\end{lemma}
Finally, it will be useful in our proofs to consider yet another variant of the isolated vertices dynamics that acts on a particular tiling with a fixed configuration in its exterior.
These chains correspond to the $S_k^\tau$'s from Section \ref{section:blocks}
for the tiled dynamics $\PIT$.
 
\bigskip\noindent
\textbf{Conditional isolated vertices tiled dynamics} $\PITC{k}{\tau}$\textbf{. }
For each $k=1,\dots,m$ and each fixed configuration $\tau$ in $B_k^c$, we consider the Markov chain with transition matrix $\PITC{k}{\tau}$ 
and state space $\PC(B_k)$, that
if $\sigma_t \in \PC(B_k)$, 
then $\sigma_{t+1} \in \PC(B_k)$ is obtained as follows:
\begin{enumerate}
	\item Add each monochromatic edge in $E$ (according to $\sigma_t \cup\tau$) independently with probability $p$;
	\item Assign to each \textit{isolated} vertex in $B_k$ independently a new spin from $\{1,\dots,q\}$ u.a.r.; 
	\item Remove all edges to obtain $\sigma_{t+1}$.
\end{enumerate}

\subsection{Proof of Theorem \ref{thm:intro:sw}}
\label{subsection:proofs}

Let 
\[\lambda_{\textrm{min}} = \min_{k=1,\dots,m} \min_{\tau \in \PC(B_k^c)} \lambda(\PITC{k}{\tau}).\] 
(Recall that $\PC(B_k^c)$ is the set of valid configurations of $B_k^c$ and $\PITC{k}{\tau}$
is the conditional isolated vertex tiled dynamics on $B_k$ with $\tau$ as the fixed configuration in the exterior of $B_k$.)
We prove the following two lemmas that, together 
with 
Lemmas \ref{lemma:sw:gap-ineq-1} and \ref{lemma:sw:gap-ineq-2}
and the results in Section \ref{section:blocks},
imply Theorem \ref{thm:intro:sw}.

\begin{lemma} \
	\label{lemma:sw:reversibility}
	\begin{enumerate}
		\item[(i)] $\PI$ and $\PIT$ are reversible w.r.t.\ $\pi$ and positive semidefinite.
		\item[(ii)] For all $k=1,\dots,m$ and $\tau \in \PC(B_k^c)$, $\PITC{k}{\tau}$ is reversible w.r.t.\ $\pi(\cdot | \tau)$ and positive semidefinite.
	\end{enumerate}
\end{lemma}

\begin{lemma}
	\label{lemma:local-gap}
	$\lambda_{{\rm min}} \ge \frac{1}{7} {\e}^{-2 \beta dL^d}$.
\end{lemma}


\begin{proof}[Proof of Theorem \ref{thm:intro:sw}]
 By Lemmas \ref{lemma:sw:gap-ineq-1} and \ref{lemma:sw:gap-ineq-2},
 $$\lambda(\PSW) \ge \lambda(\PI) \ge \lambda(\PIT).$$
 $\PIT$ is a tiled block dynamics.
 Indeed, if $\tau$ is the configuration in $B_k^c$, then the configuration in $B_k$ is updated with a step of the ergodic Markov chain $\PITC{k}{\tau}$.
 By Lemma \ref{lemma:sw:reversibility}, $\PIT$ is reversible w.r.t.\ $\pi$ and positive semidefinite. 
 Lemma \ref{lemma:sw:reversibility} also implies that $\PITC{k}{\tau}$ is reversible w.r.t.\ $\pi(\cdot|\tau)$ and positive semidefinite, for all $k=1,\dots,m$ and $\tau \in \PC(B_k^c)$.
 Hence, by Lemma \ref{lemma:parallel-block}
 $$\lambda(\PIT) \ge \lambda_{{\rm min}} \lambda(\mathcal{B}_\mathcal{D}).$$
 By Lemma
 \ref{lemma:heat-bath-parallel-block},
 when $L$ is a sufficiently large constant (independent of $n$), SSM implies that
 $\lambda(\mathcal{B}_\mathcal{D}) \ge 1/8$. Moreover, by Lemma \ref{lemma:local-gap}, 
 $\lambda_{{\rm min}} \ge \frac{1}{7} {\e}^{-2 \beta dL^d}$. Then
 \[\lambda(\PSW) \ge \frac{1}{56} {\e}^{-2 \beta dL^d}, \]
 and the result follows from the fact that $L=O(1)$.
\end{proof}

The rest of this section is organized as follows.
The proofs of Lemmas \ref{lemma:sw:gap-ineq-1}, \ref{lemma:sw:gap-ineq-2} and \ref{lemma:sw:reversibility} use a common representation of 
the Markov chains $\PSW$, $\PI$ and $\PIT$ which
we introduce in Section \ref{subsubsection:sw:common}.
The actual proofs of these lemmas are provided in Section \ref{subsection:sw:comparison}.
The proof of Lemma \ref{lemma:local-gap} is provided in Section \ref{subsection:sw:local-gap} and
crucially uses
the fact that by design the $d$-dimensional cubes of side length $L-1$
in each tiling do not interact with each other.

\subsubsection{Common representation}
\label{subsubsection:sw:common}

We provide here a decomposition of the transition matrices
$\PSW$, $\PI$ and $\PIT$ as products of simpler matrices, which will be used in our proofs of Lemmas \ref{lemma:sw:gap-ineq-1}, \ref{lemma:sw:gap-ineq-2} and \ref{lemma:sw:reversibility}.   
We are able to do this because the steps of these chains
all include a ``lifting'' substep to a \textit{joint} configuration space $\JC \subset \PC \times 2^E$, where configurations consist of a spin assignment to the vertices together with a subset of the edges of $G$. 		
The joint Edwards-Sokal measure $\nu$ on $\JC$ is given by
\[\nu(A,\sigma) = p^{|A|}(1-p)^{|E \setminus A|} \1 (A \subseteq E(\sigma)),\]
where 
$p=1-{\e}^{-\beta}$,
$A \subset E$,
$\sigma \in \PC$ and
$E(\sigma)$ denotes the set of monochromatic edges of $E$ in $\sigma$ \cite{ES}. 	

Let $T$ be the $|\PC| \times |\JC|$ matrix indexed by Potts and joint configurations given by:
\[T(\sigma,(A,\tau)) = \1(\sigma=\tau)\1(A \subseteq E(\sigma)) p^{|A|} (1-p)^{|E(\sigma)\setminus A|},\]
where $\sigma \in \PC$ and $(A,\tau) \in \JC$.
The matrix $T$ corresponds to adding each monochromatic edge of $E$ in $\sigma$ independently with probability $p$, as in step~1 of the SW dynamics, and defines an operator from $L_2(\JC,\nu)$ to $L_2(\PC,\pi)$. It is straightforward to check that its adjoint operator $T^*:L_2(\PC,\pi) \rightarrow L_2(\JC,\nu)$ is given by the $|\JC| \times |\PC|$ matrix 
$$T^*((A,\tau),\sigma) = \1(\tau = \sigma),$$
with $(A,\tau) \in \JC$ and $\sigma \in \PC$.
$T^*$ corresponds to step 3 of the SW dynamics.
Finally, let $R$ be a $|\JC| \times |\JC|$ matrix indexed by joint configurations such that
\[R((A,\sigma),(B,\tau)) = \1(A = B)\1(A \subseteq E(\sigma) \cap  E(\tau)) \cdot q^{-c(A)},\]
where $c(A)$ is the number of connected components of $(V,A)$ and $(A,\sigma),(B,\tau) \in \JC$. The matrix $R$ corresponds to assigning a new spin from $\{1,\dots,q\}$ u.a.r.\ to each connected component of $(V,A)$ independently as in step 2 of the SW dynamics.
Hence,
we get $\PSW = TRT^*$.
This useful decomposition of the SW dynamics was discovered first
in \cite{Ullrich1,Ullrich2,Ullrich4}
and has already been used in other comparison arguments involving the SW dynamics (see, e.g., \cite{BSmf,GL}).
	
The following $|\JC| \times |\JC|$ matrices 
allow us to obtain similar decompositions for $\PI$ and $\PIT$. For $(A,\sigma),(B,\tau) \in \JC$, let
\begin{align}
Q((A,\sigma),(B,\tau)) &= \1(A = B)\1(A \subseteq E(\sigma) \cap E(\tau)) \1(\sigma(V\setminus\mathcal{I}(A))=\tau(V\setminus\mathcal{I}(A))) \cdot q^{-|\mathcal{I}(A)|}\notag\\
Q_{k}((A,\sigma),(B,\tau)) &= \1(A = B)\1(A \subseteq E(\sigma) \cap E(\tau)) \1(\sigma(V\setminus\mathcal{I}_{k}(A))=\tau(V\setminus\mathcal{I}_{k}(A))) \cdot q^{-|\mathcal{I}_{k}(A)|}\notag
\end{align}
where $\mathcal{I}(A)$, $\mathcal{I}_{k}(A)$ denote the sets of isolated vertices in $V$ and $B_k$, respectively. Then, the following facts follow straightforwardly from the definition of these matrices:
\begin{fact} \
	\begin{enumerate}
	\item[(i)] $\PI = TQT^*$;
	\item[(ii)] $\PIT = \frac{1}{m} \sum_{k=1}^m TQ_kT^*$.
	\end{enumerate}
\end{fact}


\subsubsection{Proofs of Lemmas \ref{lemma:sw:gap-ineq-1}, \ref{lemma:sw:gap-ineq-2} and \ref{lemma:sw:reversibility}}
\label{subsection:sw:comparison}

In this subsection we provide our proofs of Lemmas \ref{lemma:sw:gap-ineq-1}, \ref{lemma:sw:gap-ineq-2} and \ref{lemma:sw:reversibility}, all of which use the common representation of the transition matrices
$\PSW$, $\PI$ and $\PIT$ introduced in Section \ref{subsubsection:sw:common},
as well as the analytic tools briefly reviewed in Section \ref{subsection:prelim:analytic-tools}.

\begin{proof}[Proofs of Lemmas \ref{lemma:sw:gap-ineq-1} and \ref{lemma:sw:gap-ineq-2}]
The matrix $R$ is symmetric and $\nu(A,\sigma) = \nu(A,\tau)$ for all $A \subset E$ and $\sigma, \tau \in \PC$ compatible with $A$;
hence $R$ is reversible w.r.t.\ the joint measure $\nu$ and $R = R^*$. The same holds for $Q$ and $Q_k$ for all $k=1,\dots,m$.
Moreover, since the matrices $R$, $Q$ and $Q_k$ assign spins u.a.r.\ to components of a joint configuration,
we deduce the following. 
\begin{fact}\
\label{fact:sw:joint-space}
\begin{enumerate} 
	\item[(i)] $R$, $Q$ and $Q_k$ define self-adjoint idempotent operators from $L_2(\JC,\nu)$ to $L_2(\JC,\nu)$. 
	\item[(ii)] $R = QRQ$ and $Q = Q_kQQ_k$.
\end{enumerate}	
\end{fact}
\noindent
Using this fact and the definition of the adjoint operator we get that for any $f\in\R^{|\PC|}$
\begin{align}
\inner{f}{\PSW f}{\pi} 
&= \inner{f}{T R T^* f}{\pi} 
= \inner{f}{T QRQ T^* f}{\pi} 
= \inner{QT^*f}{RQ T^* f}{\nu} \notag \\
&\le \inner{QT^*f}{Q T^* f}{\nu} 
= \inner{f}{T Q^2 T^* f}{\pi}
= \inner{f}{\PI f}{\pi}, \label{eq:sw:first-comp}
\end{align}
where the inequality follows from (\ref{eq:sw:contraction}). 
Similarly, for any $f\in\R^{|\PC|}$
\begin{align}
\inner{f}{\PI f}{\pi} 
&= \inner{f}{T Q T^* f}{\pi} 
= \inner{f}{T Q_{k} Q Q_{k} T^* f}{\pi} 
= \inner{Q_{k}T^*f}{Q Q_{k} T^* f}{\nu} \notag\\
&\le \inner{Q_{k}T^*f}{Q_{k} T^* f}{\nu} 
= \inner{f}{T Q_{k}^2 T^* f}{\pi} = \inner{f}{T Q_{k} T^* f}{\pi}. \notag
\end{align}
Since this holds for every $k$, we get
\begin{equation}
\label{eq:sw:second-comp}
\inner{f}{\PI f}{\pi} \le \frac{1}{m}\sum_{k=1}^m\inner{f}{T Q_k T^* f}{\pi} = \inner{f}{\PIT f}{\pi}.
\end{equation}
Putting (\ref{eq:sw:first-comp}) and  (\ref{eq:sw:second-comp}) together
we get 
$$\inner{f}{\PSW f}{\pi} \le \inner{f}{\PI f}{\pi} \le \inner{f}{\PIT f}{\pi}.$$ 
By Fact \ref{fact:sw:joint-space}, $R^2 = R = R^*$ and so $\inner{f}{\PSW f}{\pi} = \inner{RT^*f}{RT^* f}{\pi} \ge 0$. Hence, the matrices $\PSW$, $\PI$ and $\PIT$ are all positive semidefinite.
Then, from the definition of the Dirichlet form and (\ref{eq:sw:gap}), we get 		
$$\lambda(\PSW) \ge \lambda(\PI) \ge \lambda(\PIT),$$
as claimed.
\end{proof}

\begin{proof}[Proof of Lemma \ref{lemma:sw:reversibility}]
	Fact \ref{fact:sw:joint-space} implies that 
	$\boldsymbol{I}_{\textsc{sw}}^* = (TQT^*)^* =  \PI$ and
	$\boldsymbol{I}_{\mathcal{D}}^* = \frac{1}{m} \sum_{k=1}^m (TQ_kT^*)^*  = \PIT$.
    Hence $\PI, \PIT$ define self-adjoint operators
	from $L_2(\PC,\pi)$ to $L_2(\PC,\pi)$ and so $\PI, \PIT$ are reversible w.r.t.\ $\pi$. 
	Moreover,
	$Q^2 = Q = Q^*$ by Fact \ref{fact:sw:joint-space} and thus
	$\inner{f}{\PI f}{\pi} = \inner{QT^*f}{QT^* f}{\nu} \ge 0$. Therefore, $\PI$ is positive semidefinite.
	Similarly, we obtain that $\PIT$ is positive semidefinite, which concludes the proof of part (i) of the lemma.
	
	For part (ii), observe that
	by definition 
	$\PITC{k}{\tau}(\sigma,\sigma') = TQ_kT^*(\sigma \cup \tau,\sigma' \cup \tau)$ for all $\sigma,\sigma' \in \PC(B_k)$ and $\tau \in \PC(B_k^c)$. Since $TQ_kT^* = (TQ_kT^*)^*$ by Fact \ref{fact:sw:joint-space},  $TQ_kT^*$ is reversible w.r.t.\ $\pi$. Hence,
	\begin{align}
	\pi (\sigma \cup \tau) TQ_kT^*(\sigma \cup \tau,\sigma' \cup \tau) &= \pi (\sigma' \cup \tau) TQ_kT^*(\sigma' \cup \tau,\sigma \cup \tau) \notag\\
	\pi (\sigma \mid \tau) \PITC{k}{\tau}(\sigma,\sigma') &= \pi (\sigma' \mid \tau) \PITC{k}{\tau}(\sigma' ,\sigma) \notag
	\end{align}
    and $\PITC{k}{\tau}$ is reversible w.r.t.\ $\pi (\cdot | \tau)$.
	Finally, for $f \in \R^{|\PC(B_k)|}$ let $\hat{f} \in \R^{|\PC|}$ be such that
	$\hat{f}(\sigma \cup \tau) = f(\sigma)$ for all $\sigma \in \PC(B_k)$ and $\tau \in \PC(B_k^c)$.
	Then,
	\begin{align*}
		\inner{f}{\PITC{k}{\tau}f}{\pi(\cdot|\tau)}
		&=  \sum_{\sigma,\sigma' \in \PC(B_k)} f(\sigma) f(\sigma') \PITC{k}{\tau}(\sigma,\sigma') \pi(\sigma\mid\tau) \\
		&= \sum_{\tau \in \PC(B_k^c)}  \sum_{\sigma,\sigma' \in \PC(B_k)} \hat{f}(\sigma \cup \tau) \hat{f}(\sigma' \cup \tau)TQ_kT^*(\sigma \cup \tau,\sigma' \cup \tau) \pi(\sigma \cup \tau) \\
		&= \inner{\hat{f}}{TQ_kT^*\hat{f}}{\pi} = \inner{Q_kT^*\hat{f}}{Q_kT^*\hat{f}}{\pi} \ge 0,
	\end{align*}
	where in the last equality we used that $Q_k = Q_k^2 = Q_k^*$ which follows from Fact \ref{fact:sw:joint-space}.
    Thus, $\PITC{k}{\tau}$ is positive semidefinite for all $1 \le k \le m$ and $\tau \in \Omega(B_k^c)$.
\end{proof}

\subsubsection{Proof of Lemma \ref{lemma:local-gap}}
\label{subsection:sw:local-gap}

In this subsection we prove Lemma \ref{lemma:local-gap} by showing that
$\lambda(\PITC{k}{\tau}) \ge \frac{1}{7} {\e}^{-2\beta dL^d}$ for all $k=1,\dots,m$ and $\tau \in \PC(B_k^\tau)$.
As mentioned earlier, our proof uses
the fact in each tiling 
the small $d$-dimensional cubes do not interact with each other.
Hence, $\PITC{k}{\tau}$ is a \textit{product} Markov chain
where each component acts 
on exactly one of the $d$-dimensional cubes of the tiling $B_k$. 
The spectral gap
of $\PITC{k}{\tau}$ is then 
given by the smallest spectral gap of any component.
The spectral gap of any component can be bounded using a crude coupling argument, since each component
acts on a set of constant volume. We proceed to formalize these ideas.

The following linear algebra fact about the spectrum of a product Markov chain will be used in the proof of Lemma \ref{lemma:local-gap}.

\begin{lemma}
	\label{lemma:sw:prod}
	Let $S_1,\dots,S_t$ be a finite spaces, and call $\mathcal{C}=S_1\times\cdots\times S_t$ their cartesian product. For $i=1,\dots,t$ let $P_i$  be the transition matrix of an ergodic Markov chain acting on $S_i$ reversible w.r.t.\ a probability measure $\varphi_i$ on $S_i$. Let $P=\prod_{i=1}^tP_i$ be the matrix given by
	$$
	P(x,y) = \prod_{i=1}^tP_i(x_i,y_i),
	$$ 
	where $x=(x_1,\dots,x_t)\in\mathcal{C}$ and $y=(y_1,\dots,y_t)\in\mathcal{C}$, $x_i\in S_i$, and $y_i\in S_i$. Then,	
	$\lambda(P)= \min\limits_{i=1,\dots,t}\lambda(P_i)$.
\end{lemma}

We provide next the proof of Lemma \ref{lemma:local-gap}.

\begin{proof}[Proof of Lemma \ref{lemma:local-gap}]
	Recall that 
	\[\lambda_{\textrm{min}} = \min_{k=1,\dots,m} \min_{\tau \in \Omega(B_k^c)} \lambda(\PITC{k}{\tau}).\]
	We claim that $\PITC{k}{\tau}$ is a product chain.
	Indeed,
	if $B_k^{(1)},\dots,B_k^{(l_k)}$ are the $d$-dimensional cubes that form the tiling $B_k$ and $\PITC{kj}{\tau}$ is the isolated vertices dynamics acting on $B_k^{(j)}$ 
	(with the boundary condition induced by $\tau$),
	then for $\sigma,\sigma' \in \PC(B_k)$,
	\[\PITC{k}{\tau}(\sigma,\sigma') = \prod_{j=1}^{l_k} \PITC{kj}{\tau} ( \sigma(B_k^{(j)}), \sigma'(B_k^{(j)})).\]
	Hence, by Lemma  \ref{lemma:sw:prod}
	\[\lambda(\PITC{k}{\tau}) = \min_{j=1,\dots,l_k} \lambda(\PITC{kj}{\tau}).\]
	
	We bound $\lambda(\PITC{kj}{\tau})$ via a crude coupling argument.
	Since $|B_k^{(j)}| \le L^d$,
	the probability that in the first step of $\PITC{kj}{\tau}$ every vertex is isolated is $(1-p)^{K}$, where $K \le 2 d L^d$ is the number of edges incident to $B_k^{(j)}$.
	Starting from two arbitrary configurations in $B_k^{(j)}$, if all vertices become isolated in both configurations, then we can couple them with probability $1$. Hence, we can couple two arbitrary configurations in one step with probability at least $(1-p)^{2 d L^d}$. Therefore, the probability that the two copies have not couple after $4(1-p)^{-2 d L^d}$ steps is at most $1/4$ by Markov's inequality.
	Then, the mixing time of $\PITC{kj}{\tau}$ 
	is at most $4(1-p)^{-2 d L^d} = 4 {\e}^{2 \beta d L^d}$
	for each $k=1,\dots, m$, $\tau \in \PC(B_k^c)$ and $j  = 1,\dots,l_k$. Consequently, 
	$\lambda(\PITC{k}{\tau}) \ge \frac{1}{7} {\e}^{-2 \beta d L^d}$
    by (\ref{prelim:eq:gap-lower}).
\end{proof}

For completeness, we also provide here a proof of Lemma \ref{lemma:sw:prod}.

\begin{proof}[Proof of Lemma \ref{lemma:sw:prod}]
	$P$ is reversible w.r.t.\ $\varphi=\otimes_{i=1}^n\varphi_i$.
	Moreover, if $\{f^{(i)}_j, l^{(i)}_j, j=1,\dots,|S_i|\}$ denote eigenfunctions and eigenvalues of $P_i$, respectively, then 
	$$
	F_k(x)=\prod_{i=1}^tf^{(i)}_{k_i}(x_i),\quad l_k=\prod_{i=1}^t l^{(i)}_{k_i},
	$$ 
	are the eigenfunctions and eigenvalues of $P$, where $k=(k_1,\dots,k_t)$, and $k_i=1,\dots,|S_i|$, for all $i=1,\dots,t$. 
	To see this, note that $\{F_k\}$, $k=(k_1,\dots,k_t)$, form an orthogonal basis in $L^2(\mathcal{C},\varphi)$, such that $PF_k = l_k F_k$. This implies that $F_k,l_k$ are the eigenfunctions and eigenvalues of $P$. 
	
	Now, suppose that $l^{(i)}_2$ is the eigenvalue $l^{(i)}_j\neq 1$ with maximal absolute value for all $i$, so that 
	$\lambda(P_i)=1-l^{(i)}_2$. Then, by taking all $l^{(i)}_{k_i}=1$ except for the one index $i_0$ and by setting $l_k=l^{(i_0)}_2$
	one has $\lambda(P) = 1-\max_{i}  l^{(i)}_2$. 
\end{proof}

\section{SSM and general block dynamics}
\label{section:partition}

In this section we
use our results for the tiled block dynamics in Section \ref{section:blocks} to deduce
a tight spectral gap bound for general heat-bath block dynamics.
Let $V \subset \Z^d$ be a $d$-dimensional cube of volume $n$,
$G=(V,E)$ the induced subgraph and $\psi$ a fixed boundary condition on $\boundary V$. 

Let $\mathcal{A} = \{A_1,\dots,A_r\}$
be a collection of blocks
such that $A_i \subset V$ and $V = \cup_i A_i$. Let $\mathcal{B}_\mathcal{A}$ be the transition matrix of 
the heat-bath block dynamics w.r.t.\ $\mathcal{A}$.
Recall that
given
a configuration $\sigma_t \in \Omega$ at time $t$
a step of the heat-bath block dynamics picks a block $A_i$ u.a.r.\ and updates the configuration in $A_i$ with a sample from $\mu^\psi(\cdot|\sigma_t(V\setminus A_i))$.
We prove here that $\lambda(\mathcal{B}_\mathcal{A}) = \Omega(r^{-1})$ whenever SSM holds. That is, we establish Theorem \ref{thm:intro:general:block} from the introduction.

In the proof of this theorem we 
relate the spectral gap of $\mathcal{B}_\mathcal{A}$ to that of the following block dynamics.
Let $V_{\rm e}$ and $V_{\rm o}$ be the set of all even and all odd vertices of $V$, respectively.
A vertex is even (resp., odd) if its coordinate sum in $\,\Z^d$ is even (resp., odd). 
Let $\mathcal{B}_{\rm eo}$ be the heat-bath block dynamics w.r.t.\ $\{V_{\rm e},V_{\rm o}\}$. 
A crucial part of the proof of Theorem \ref{thm:intro:general:block} is the following.
\begin{lemma}
	\label{lemma:block:even-odd}
	SSM implies that  $\lambda(\mathcal{B}_{\rm eo}) = \Omega(1)$.
\end{lemma}	
\noindent
The other key ingredients in the proof of Theorem \ref{thm:intro:general:block} are two properties of the
variance functional: monotonicity and tensorization. (Recall that for $A \subseteq V$, $K_A$ denotes the matrix that corresponds to the heat-bath update in $A$ and that we use $\mathcal{E}_A$ for the Dirichlet form of $K_A$.)

\begin{fact} 
	\label{fact:var-monotonicity}
	Let $A \subseteq B \subseteq V$. Then, for any $f \in \R^{|\Omega|}$,
	$\df{A}{f} \le \df{B}{f}.$
\end{fact}

\begin{fact} 
	\label{fact:var-tensorization}
	Let $U = \cup U_i \subseteq V$
	such that
	$K_{U_i} K_{U_j} = K_{U_j} K_{U_i}$ for all $i \neq j$.
	Then, for any $f \in \R^{|\Omega|}$
	$$\df{U}{f} \le \sum_{i} \df{U_i}{f}.$$
\end{fact}

We are now ready to prove Theorem \ref{thm:intro:general:block}.

\begin{proof}[Proof of Theorem \ref{thm:intro:general:block}]
	For any $f \in \R^{|\Omega|}$, we have
	$\df{\mathcal{B}_\mathcal{A}}{f} = \frac{1}{r} \sum_{i=1}^r \df{A_i}{f}$.
    By Fact \ref{fact:var-monotonicity}, if $A_i' \subset A_i$, then $\df{A_i}{f} \ge \df{A_i'}{f}$.
    Thus, we may assume without loss of generality that $\mathcal{A}$ is a partition of $V$. 
    Fact \ref{fact:var-monotonicity} also implies
 	\begin{align*}
 	\df{A_i}{f} \ge \frac{\df{A_i \cap V_{\rm e}}{f} + \df{A_i \cap V_{\rm o}}{f}}{2} 
	\end{align*}
	Hence,
	$$
	\df{\mathcal{B}_\mathcal{A}}{f}  \ge  \frac{1}{r} \sum_{i=1}^r \frac{\df{A_i \cap V_{\rm e}}{f} + \df{A_i \cap V_{\rm o}}{f}}{2}.
	$$
	For $i \neq j$, $\dist(A_i \cap V_{\rm e} , A_j \cap V_{\rm e}) \ge 2$, since by assumption $A_i \cap A_j = \emptyset$. Then, 
	$$K_{A_i\cap V_{\rm e}}K_{A_j\cap V_{\rm e}} = K_{A_j\cap V_{\rm e}}K_{A_i\cap V_{\rm e}}$$
	and 
	\begin{align*}
	\sum_{i=1}^r \df{A_i \cap V_{\rm e}}{f} 
	\ge \df{V_{\rm e}}{f}
	\end{align*}
	by Fact \ref{fact:var-tensorization}.
	Similarly, we get $
	\sum_{i=1}^r \df{A_i \cap V_{\rm o}}{f} \ge \df{V_{\rm o}}{f}$. Hence,
	$$\df{\mathcal{B}_\mathcal{A}}{f} \ge \frac{\df{V_{\rm e}}{f} +\df{V_{\rm o}}{f} }{2r} = \frac{1}{r}\df{\mathcal{B}_{\rm eo}}{f}.$$
	Since $\mathcal{B}_\mathcal{A}$ and $\mathcal{B}_{\rm eo}$ are both positive semidefinite we get
	$\lambda(\mathcal{B}_\mathcal{A}) \ge \frac{1}{r}\lambda(\mathcal{B}_{\rm eo})$ by (\ref{eq:sw:gap}). 
	The result follows from Lemma~\ref{lemma:block:even-odd}.
\end{proof}

To prove Lemma \ref{lemma:block:even-odd} 
we use
our results for tiled block dynamics from Section \ref{section:blocks}. 
In particular,
we consider the tiled block dynamics that picks one tiling $B_i$ from $\mathcal{D} = \{B_1,\dots,B_m\}$ u.a.r.\ and with probability $1/2$ performs a heat-bath update in $B_i \cap V_{\rm e}$ and otherwise updates $B_i \cap V_{\rm o}$.
The restriction of this tiled block dynamics to each $B_i$ is not a product Markov chain, as it was the case in the previous application of our technology to the SW dynamics in Section \ref{section:sw}.
Hence, we cannot hope to use Lemma \ref{lemma:sw:prod} for product Markov chains directly. To work around this difficulty we consider systematic scan variants of the restricted chains.

\begin{proof}[Proof of Lemma \ref{lemma:block:even-odd}]
	For ease of notation let $\mu = \mu^\psi$.
	Let $\BT$ be the transition matrix of 
	the tiled variant of $\mathcal{B}_{\rm eo}$ that given a configuration $\sigma_t$ proceeds as follows:
	\begin{enumerate}
		\item Pick $j \in \{1,...,m\}$ u.a.r.;
		
		\item With probability $1/2$ update the spins of $V_{\rm e} \cap B_j$ with a sample from 
		$\mu(\cdot|\sigma_t(V \setminus (V_{\rm e} \cap B_j)))$;
		
		\item Otherwise, update the configuration in $V_{\rm o} \cap B_j$ with a sample from $\mu(\cdot|\sigma_t(V \setminus (V_{\rm o} \cap B_j)))$.
	\end{enumerate}
	This chain is reversible w.r.t.\ $\mu$ and ergodic; the latter follows directly from the assumption that the heat-bath Glauber dynamics is ergodic (see Section \ref{subsection:mc}). 
	
	By Fact \ref{fact:var-monotonicity},  
	$\df{V_{\rm e}}{f} \ge \df{V_{\rm e} \cap B_j}{f}$ and
	$\df{V_{\rm o}}{f} \ge \df{V_{\rm o} \cap B_j}{f}$
	for any $f \in \R^{|\Omega|}$. Thus,
	\begin{align}
	\df{\mathcal{B}_{\rm eo}}{f} = \frac{\df{V_{\rm e}}{f}+\df{V_{\rm o}}{f}}{2} &\ge \frac{\df{V_{\rm e} \cap B_j}{f}+\df{V_{\rm o} \cap B_j}{f}}{2}  \notag\\
	&\ge \frac{1}{m}\sum_{j=1}^m \frac{\df{V_{\rm e} \cap B_j}{f}+\df{V_{\rm o} \cap B_j}{f}}{2} = \df{\BT}{f}.
	\end{align}	
	Since both $\BT$ and $\mathcal{B}_{\rm eo}$ are averages of positive semidefinite matrices (see Fact \ref{fact:var-heat-bath-positive}), they are also positive semidefinite
	and so 
	$$\lambda(\mathcal{B}_{\rm eo}) \ge \lambda(\BT).$$
	
	We bound next $\lambda(\BT)$. For each $j=1,\dots,m$ and each configuration $\tau \in \Omega(B_j^c)$, 
	we consider the Markov chain with transition matrix $\BT_j^\tau$ 
	whose state space 
	is the set $\Omega^\tau(B_j)$ of valid configurations in $B_j$ given that $\tau$ is the configuration in $B_j^c$.
	Given a configuration $\sigma_t$, this chain obtains $\sigma_{t+1}$ as follows:
	\begin{enumerate}
		\item With probability $1/2$ update the spins of $V_{\rm e} \cap B_j$ with a sample from $\mu(\cdot|\sigma_t(B_j \setminus (V_{\rm e} \cap B_j)),\tau)$;
		\item Otherwise, update the configuration in $V_{\rm o} \cap B_j$ with a sample from $\mu(\cdot|\sigma_t(B_j \setminus (V_{\rm o} \cap B_j)),\tau)$.
	\end{enumerate}
	It is straightforward to check that this chain is ergodic and reversible w.r.t.\ $\varphi = \mu(\cdot | \tau)$. 
	Moreover, $\BT_j^\tau$ is positive semidefinite since it is an average of heat-bath updates (see Fact \ref{fact:var-heat-bath-positive}). (Observe that the Markov chains $P_j^\tau$'s correspond to the $S_j^\tau$'s from Secion \ref{section:blocks}.)
	
	Let 
	$$\lambda_{\rm min} = \min_{j=1,\dots,m} \min_{\tau \in \Omega(B_j^c)} \lambda(\BT_j^\tau).$$
	By Lemma \ref{lemma:parallel-block},
	$\lambda(\BT) \ge \lambda_{\rm min}  \lambda(\mathcal{B}_\mathcal{D})$
	and, by Lemma \ref{lemma:heat-bath-parallel-block},
	$ \lambda(\mathcal{B}_\mathcal{D}) \ge \frac{1}{8}$,
	provided $L$ is a large enough constant independent of $n$ and that there is SSM. Hence,
	\begin{equation}
	\label{eq:h:min}
	\lambda(\mathcal{B}_{\rm eo}) \ge \frac{\lambda_{\rm min}}{8}.
	\end{equation}
	
	We show next that
	$\lambda_{\rm min} = \Omega(1)$
	by bounding $\lambda(\BT_j^\tau)$ for each $j$ and $\tau$.
	Fix $j$ and $\tau$ and
	let $P_{\rm e}$ (resp., $P_{\rm o}$) be the transition matrix that corresponds 
	to 
	updating the configuration in $V_{\rm e} \cap B_j$ (resp., $V_{\rm o} \cap B_j$)
	with a new configuration distributed according to the conditional measure given the configuration in $B_j \setminus (V_{\rm e} \cap B_j)$ (resp., $B_j \setminus (V_{\rm o} \cap B_j)$) and $\tau$. 
	$P_{\rm e}$ and $P_{\rm o}$ are reversible w.r.t.\ $\varphi$ and $\BT_j^\tau = \frac{P_{\rm e}+ P_{\rm o}}{2}$.
	
	Let
	$P_\textrm{eoe} = P_\textrm{e}P_\textrm{o}P_\textrm{e}$ 
	be a systematic scan variant of $\BT_j^\tau$
	and let
	$P_\textrm{eoe}^{\textsc{l}}$
	be the ``lazy'' version of $P_\textrm{eoe}$ that with probability $7/8$ stays put and
	with probability $1/8$ proceeds like $P_{\textrm{eoe}}$; that is, $P_\textrm{eoe}^{\textsc{l}} = \frac{P_{\textrm{eoe}}+7I}{8}$.
	We show that three steps of the chain $\BT_j^\tau$ are as fast as one of $P_\textrm{eoe}^{\textsc{l}}$. For this, note that		
	\begin{equation}
	\label{eq:pj-cube}
	(\BT_j^\tau)^3 = \frac{1}{8}( P_\textrm{e}P_\textrm{o}P_\textrm{e} + P_\textrm{e}^3 + P_\textrm{o}^3 + P_\textrm{e}^2P_\textrm{o}+P_\textrm{o}P_\textrm{e}^2+P_\textrm{o}^2P_\textrm{e}+P_\textrm{e}P_\textrm{o}^2+P_\textrm{o}P_\textrm{e}P_\textrm{o}).
	\end{equation}
	Each of the terms in the right hand side of (\ref{eq:pj-cube}) is at most $\inner{f}{f}{\varphi}$ by (\ref{eq:sw:contraction}).
	Thus,
	\[\inner{f}{{(\BT_j^\tau)^3} f}{\varphi} 
	\le \frac{1}{8}\inner{f}{P_\textrm{e}P_\textrm{o}P_\textrm{e}f}{\varphi}+\frac{7}{8}\inner{f}{f}{\varphi}
	=\inner{f}{P_\textrm{eoe}^{\textsc{l}}f}{\varphi}.\]
	By Fact \ref{fact:var-heat-bath-positive} the matrices $P_\textrm{e}$ and $P_\textrm{o}$ are positive semidefinite, and thus $\BT_j^\tau$, $(\BT_j^\tau)^3$, $P_\textrm{eoe}$ and $P_\textrm{eoe}^{\textsc{l}}$ are also positive semidefinite. Then,
	$$\lambda({(\BT_j^\tau)^3}) \ge \lambda(P_\textrm{eoe}^{\textsc{l}}).$$
	
	Since $x^3-3x+2 \ge 0$ for $|x| \le 1$, we have
	$3\lambda(\BT_j^\tau) \ge \lambda((\BT_j^\tau)^3)$. 
	Moreover, 
	$$\mathcal{E}_{P_\textrm{eoe}^{\textsc{l}}}(f,f) = \inner{f}{(I-P_\textrm{eoe}^{\textsc{l}})f}{\varphi} = \frac{1}{8} \mathcal{E}_{P_\textrm{eoe}} (f,f),$$
	 and so $\lambda(P_\textrm{eoe}^{\textsc{l}}) = \frac{1}{8} \lambda(P_\textrm{eoe})$. Hence,
	\begin{equation}
	\label{eq:lazy-bound}
	\lambda({\BT_j^\tau}) \ge \frac{1}{24} \lambda(P_\textrm{eoe}).
	\end{equation}
	
	We bound next $\lambda(P_\textrm{eoe})$.
	Let $B_{j}^{(1)},B_{j}^{(2)},\dots,B_{j}^{(l)}$ be the $d$-dimensional cubes of volume at most $L^d$ that form the tiling $B_j$.
	
	For $k=1,\dots,l$
	let $P_{\rm e}^{(k)}$ and $P_{\rm o}^{(k)}$
	be the
	$|\Omega^\tau(B_j^{(k)})| \times |\Omega^\tau(B_j^{(k)})|$
	transition matrices 
	that correspond to a heat-bath update on
	$V_{\rm e} \cap B_j^{(k)}$ and $V_{\rm o} \cap B_j^{(k)}$, respectively. 
	Let $P_\textrm{eoe}^{(k)} = P_{\rm e}^{(k)}P_{\rm o}^{(k)}P_{\rm e}^{(k)}$.
	For $\sigma,\sigma' \in \Omega^{\tau}(B_j)$, we have
	$$
	P_{\textrm{eoe}}(\sigma,\sigma') = \prod_{k=1}^{l} 	P^{(k)}_{\textrm{eoe}}(\sigma(B_j^{(k)}),\sigma'(B_j^{(k)})).
	$$
	Moreover, $P^{(k)}_{\textrm{eoe}}$ ergodic and reversible w.r.t.\ the probability measure induced in $B_j^{(k)}$ by $\varphi$.
	The former follows from the fact that by assumption the heat-bath dynamics on $B_j^{(k)}$ is ergodic; see Section \ref{subsection:prelim:markov-chains}.
	Thus, Lemma \ref{lemma:sw:prod} implies 
	\begin{equation}
	\label{eq:s:min}
	\lambda(P_{\textrm{eoe}}) = \min_{k=1,\dots,l} \lambda(P_{\textrm{eoe}}^{(k)}).
	\end{equation}
	We bound $\lambda(P_{\textrm{eoe}}^{(k)})$ for each $k$ with a crude coupling argument. This is sufficient because each $B_j^{(k)}$ has volume at most $L^d = O(1)$.
	For any $U \subseteq B_j^{(k)}$ and any spin configuration $\eta$ on $B_j^{(k)} \setminus U$,
	the probability of each valid configuration on $U$ given $\eta$ and $\tau$ can be crudely bounded from below by ${(q \e)}^{-\Omega( L^d)}$.
	Since $P_{\textrm{eoe}}^{(k)}$ is irreducible,
	for any pair of configurations $\sigma_0,\sigma_0'$ of $B_j^{(k)}$,
	we can go from $\sigma_0$ to $\sigma_0'$ in at most $T = q^{L^d}$ steps. 
	Therefore,
	the probability that a realization of $P_{\textrm{eoe}}^{(k)}$
	follows this sequence of updates is then at least ${(q \e)}^{- \Omega(T L^d)}$. 
	Moreover, the probability that an instance of $P_{\textrm{eoe}}^{(k)}$ that starts in $\sigma_0'$ remains at $\sigma_0'$ after $T$ steps is also at least ${(q \e)}^{- \Omega(T L^d)}$.
	Thus, there exists a coupling for the steps of $P_{\textrm{eoe}}^{(k)}$ that starting from an arbitrary pair of configurations couples in $O(1)$ steps with probability $\Omega(1)$. Consequently, $\lambda(P_{\textrm{eoe}}^{(k)}) = \Omega(1)$ for all $k$.
	This bound together with (\ref{eq:s:min}) and (\ref{eq:lazy-bound}) imply that
	$\lambda(\BT_{j}^\tau) = \Omega(1)$,
	and so
	$\lambda_{\rm min}  = \Omega(1)$.
	The result follows from (\ref{eq:h:min}).
\end{proof}

We conclude this section with the proofs of Facts \ref{fact:var-monotonicity} and \ref{fact:var-tensorization}.

\begin{proof}[Proof of Fact \ref{fact:var-monotonicity}]
	Since $A \subseteq B$, $K_B = K_AK_BK_A$. Then, for any $f \in \R^{|\Omega|}$
	$$\inner{f}{K_Bf}{\mu} = \inner{f}{K_AK_BK_Af}{\mu} = \inner{K_Af}{K_BK_Af}{\mu} \le \inner{K_Af}{K_Af}{\mu} = \inner{f}{K_Af}{\mu} ,$$
	where the inequality follows from (\ref{eq:sw:contraction}). 
	Then,  we get $\df{B}{f} \ge \df{A}{f}$.
\end{proof}

\begin{proof}[Proof of Fact \ref{fact:var-tensorization}]
	To simplify the notation, let $K_i = K_{U_i}$ and $K_{ij} = K_{U_i \cup U_j}$.
	By assumption $K_{ij} = K_i K_j = K_j K_i$; also, 
	$K_i^2 = K_i$.
	Then,
	$I-K_{i} = (I-K_{i})^2$, $(I-K_{i})(I-K_{j}) = (I-K_{j})(I-K_{i})$ and
	\begin{align*}
	\inner{f}{(I-K_{i} - K_{j} + K_{ij})f}{\mu} 
	&= \inner{f}{(I-K_{i})(I-K_{j})f}{\mu} \\
	&= \inner{(I-K_{i})(I-K_{j})f}{(I-K_{i})(I-K_{j})f}{\mu} \\
	&\ge 0.
	\end{align*}
	Hence, $\df{K_{i}}{f} + \df{K_{j}}{f} \ge \df{K_{ij}}{f}$. Applying this to $U_1$ and $U_2$ first, and then iterating we get the result. 
\end{proof}

\section{SSM and the system scan dynamics}
\label{section:ss}

Let $V \subset \Z^d$ be a finite $d$-dimensional cube of volume $n$.
Let $G=(V,E)$ be the induced subgraph and let $\psi$ be a fixed boundary condition on $\boundary V$. For ease of notation we use  $\mu$ for $\mu^\psi$.

We consider in this section the class of \textit{systematic scan} Markov chains on $G$. In a systematic scan chain there is a fixed ordering $\Ord$ of the vertices of $G$ and \textit{one} step of the chain consists of updating every $v \in V$ according to the conditional distribution at $v$ given the configuration of its neighbors and the boundary condition $\psi$, in the order specified by $\Ord$. 
We use $\GD(\Ord)$ to denote the systematic scan dynamics w.r.t.\ the ordering $\Ord$ and $S(\Ord)$ to denote its transition matrix.
Hence, if $\Ord = \{v_1,\dots,v_n\}$
$$S(\Ord) = K_{v_1}\dots K_{v_n}$$
(Recall that $K_{v_i}$ is
the transition matrix corresponding to a heat-bath update in $v_i$.)
Since each $K_{v_i}$ leaves $\mu$ invariant then $\mu$ is the equilibrium distribution of $S(\Ord)$.
In general $S(\Ord)$ is non-reversible, but one can obtain a reversible matrix by multiplicative symmetrization (see, e.g., \cite{Fill,MT}): 
\begin{equation*}\label{sysc2}
S(\Ord)S(\Ord)^*=K_{v_1}\dots K_{v_{n-1}}K_{v_n}K_{v_{n-1}}\dots K_{v_1},
\end{equation*}
which corresponds to 
the systematic scan dynamics $\GD(\Ord')$
with  $\Ord'=\{v_1,\dots,v_n,\dots,v_1\}$. 

In this section we prove three results related to the 
speed of convergence to equilibrium of systematic scan dynamics.
These results correspond to Theorems \ref{thm:intro:ss-monotone} and \ref{thm:intro:ss:eo} from the introduction.

The first of our results concerns the \textit{alternating scan dynamics},
which corresponds to the systematic scan dynamics
whose ordering consists of first all the even vertices and then all the odd ones.
In fact, we consider the multiplicative reversiblization of this dynamics as above.
More formally, 
let $EO$ be an ordering of the vertices of $V$ that first contains all even vertices and then all the odd ones.
Similarly define the ordering $EOE$, that contains all even vertices, then all the odd ones, and finally all the even ones again.
The alternating scan dynamics on $G$ correponds to the systmatic scan dynamics $\GD(EO)$.
The relaxation time of the non-reversible chain $\GD(EO)$ is given by 
\begin{equation}
\label{eq:relax:non-rev}
\tau_{\rm rel}(\GD(EO)) = \frac{1}{1-\sqrt{1-\lambda(S(EOE))}};
\end{equation}
see, e.g., \cite{Fill,MT}.
Thus, we may restrict our attention to estimating the spectral gap of the reversible Markov chain $\GD(EOE)$. 
Let $V_{\rm e}$ (resp., $V_{\rm o}$) be the set of the even (resp., the odd) vertices of $G$. Then, $S(EOE) = K_{V_{\rm e}}K_{V_{\rm o}}K_{V_{\rm e}}$.
We prove the following.

\begin{thm}
	\label{thm:ss:eo}
	SSM implies that $\lambda(S(EOE)) \ge \Omega(1)$.
\end{thm}
\noindent
We observe that Theorem \ref{thm:ss:eo} and (\ref{eq:relax:non-rev}) imply Theorem \ref{thm:intro:ss:eo} from the introduction.

For the special case of {\it monotone\/} spin systems we show that SSM implies rapid mixing of \textit{any} systematic scan dynamics.
In a monotone system for each vertex $v \in V$ there is
a linear ordering $\succeq_v$ of the spins.  
These linear orderings induce a partial order $\succeq$ over the state space.
The spin system is monotone w.r.t.\ this partial order if for every $B \subset V$ and every pair of boundary conditions $\xi_1 \succeq \xi_2$ on $\boundary B$, 
$\mu_B^{\xi_1}$ stochastically dominates $\mu_B^{\xi_2}$.
From this definition it follows that a monotone system has unique maximal and minimal configurations in the partial order $\succeq$, a fact that will be crucially used in our proofs.
Several well-known spin systems, including the Ising model and the hard-core model, are monotone systems.

For monotone systems we establish the following two theorems which together imply Theorem \ref{thm:intro:ss-monotone} from the introduction.

\begin{thm}
	\label{thm:ss:monotone}
	Let $\Ord$ be an ordering of the vertices in $V$.
	In a monotone system
	SSM implies that 
	the mixing time of $\GD(\Ord)$
	is
	$O(\log n (\log \log n)^2)$.
\end{thm}
\noindent
We emphasize that Theorem \ref{thm:ss:monotone} holds for any ordering $\Ord$
and any boundary condition $\psi$ on $\boundary V$.

Let $\mathcal{L}(\Ord)$ be the length of the longest subsequence of $\Ord$ that is a path in $G$. 
With the additional assumption that $\mathcal{L}(\Ord) = O(1)$ we can prove a slightly better bound for the mixing time of the systematic scan dynamics.

\begin{thm}
	\label{thm:ss:monotone:no-propagation}
	Let $\Ord$ be an ordering of the vertices in $V$
	such that $\mathcal{L}(\Ord) = O(1)$.
	In a monotone system SSM implies that
	the mixing time of $\GD(\Ord)$ is $O(\log n)$
	and that the spectral gap of $\GD(\Ord)$ is $\Omega(1)$. 
\end{thm}

We proceed to give proofs to these three theorems.
We start with the proof of
Theorem \ref{thm:ss:eo}, which is deduced straightforwardly from the following more general fact.

\begin{lemma}
	\label{lemma:general:eo}
	Let $S,T$ be positive semidefinite stochastic matrices, reversible w.r.t.\ $\mu$. Assume that $S$ is also idempotent. Then, for all $a\in[0,1]$
	\begin{equation}
	\label{eq:ineq-general}
	\lambda(S\,TS)\geq \lambda(aS+(1-a)T)
	\end{equation}
\end{lemma} 

\begin{proof}[Proof of Theorem \ref{thm:ss:eo}]
	Since, by Fact \ref{fact:var-heat-bath-positive},
	$K_{V_{\rm e}}$ and $K_{V_{\rm o}}$ are positive semidefinite matrices,
	and $K_{V_{\rm e}}$ is idempotent, it follows from Lemma \ref{lemma:general:eo} that
	$$\lambda(S(EOE)) \ge \lambda\left(\frac{K_{V_{\rm e}}+K_{V_{\rm o}}}{2}\right) = \lambda(\mathcal{B}_{\rm eo}),$$
	where $\mathcal{B}_{\rm eo}$ is the block dynamics considered in Section \ref{section:partition}.
	From Lemma \ref{lemma:block:even-odd} we know that $\lambda(\mathcal{B}_{\rm eo} )= \Omega(1)$ whenever SSM holds, and thus the result follows.
\end{proof}	 

\begin{proof}[Proof of Lemma \ref{lemma:general:eo}]
	Let $P=STS$. For any $f \in \R^{|\Omega|}$
	$$\inner{f}{Pf}{\mu} = \inner{f}{STSf}{\mu} = \inner{Sf}{TSf}{\mu} \ge 0,$$
	since by assumption $T$ is positive semidefinite. Hence, $P$ is positive semidefinite and $\lambda(P) = 1- \lambda_2(P)$, where $\lambda_2(P)$ is the maximal eigenvalue of $P$ different from $1$.
	By the variational principle (see (\ref{eq:sw:gap}))
	$$
	\lambda_2(P) = \max_{f:\; \mu(f)=0,\, \|f\|\leq 1}\inner{f}{Pf}{\mu},
	$$
	where $\mu(f)=\sum_{\sigma \in \Omega} f(\sigma) \mu(\sigma)$ and $\|f\|^2=\inner{f}{f}{\mu}$.
	
	Let $Q = aS+(1-a)T$ and
	let $g \in \R^{|\Omega|}$ be such that $\mu(g)=0$ and $\|g\|=1$. 
	Then
	\begin{equation}\label{boqa}
	\lambda_2(Q) =\max_{f:\; \mu(f)=0, \|f\|\leq 1}\inner{f}{Qf}{\mu}\geq  
	\inner{Sg}{QSg}{\mu} =\inner{g}{SQSg}{\mu},
	\end{equation}
	where 
	the inequality follows from the fact 
	that any $g$ with $\mu(g)=0$ and $\|g\| = 1$ satisfies 
	\begin{align*}
	&\mu(Sg)=\inner{\vec{1}}{Sg}{\mu}=\inner{S\cdot\vec{1}}{g}{\mu}=\mu(g)=0,~\textrm{and}\\
	 & \|Sg\|^2=\inner{g}{S^2g}{\mu}\leq 1
	  \end{align*}
	  by (\ref{eq:sw:contraction}).
	On the other hand, 
	\begin{equation*}
	\label{boq}
	\inner{f}{Pf}{\mu} = \inner{Sf}{TSf}{\mu}\leq \inner{Sf}{Sf}{\mu} =  \inner{f}{Sf}{\mu},
	\end{equation*}
	where the inequality follows from (\ref{eq:sw:contraction}).
	Hence, 
	$$
	\inner{g}{SQSg}{\mu} = \inner{g}{aS^3+(1-a)Pg}{\mu} =
	a\inner{g}{Sg}{\mu} + (1-a) \inner{g}{Pg}{\mu} \geq   \inner{g}{Pg}{\mu}. 
	$$
	Therefore, taking $g$ as a normalized eigenfunction corresponding to $\lambda_2(P)$ we get
	$$
	\lambda_2(Q) \geq \lambda_2(P).
	$$ 
	This proves $\lambda(STS)\geq\lambda(aS+(1-a)T)$, as desired. 
\end{proof}

\begin{remark}
	A reverse inequality for (\ref{eq:ineq-general}) also holds. 
	Indeed, using the same argument as in the proof of the estimate (\ref{eq:lazy-bound}) one has, for all $a\in(0,1)$:
	$$\lambda(STS)\leq \frac{3}{a^2(1-a)} \lambda(aS + (1-a)T).$$
\end{remark}

We provide next the proofs of
Theorems \ref{thm:ss:monotone} and \ref{thm:ss:monotone:no-propagation}.

\begin{proof}[Proof of Theorem \ref{thm:ss:monotone}]
	
	Let $\sigma^{(1)}, \dots, \sigma^{(k)}$ be spin configurations such that $\sigma^{(1)} \succeq \dots \succeq \sigma^{(k)}$.
	For any $v \in V$,
	the monotonicity of the system implies 
	that there exists a monotone coupling for updating $v$ simultaneously in $\sigma^{(1)}, \dots, \sigma^{(k)}$
	such that the resulting configurations, denoted $\sigma^{(1)}_v, \dots, \sigma^{(k)}_v$, satisfy $\sigma^{(1)}_v \succeq \dots \succeq \sigma^{(k)}_v$. 
	These local couplings
	can be straightforwardly extended to a monotone coupling for the steps of any number of copies $\GD(\Ord)$.
	Indeed,
	let $\{X_t^{(1)}\},\dots, \{X_t^{(k)}\}$ be $k$ copies of $\GD(\Ord)$ and suppose $X_t^{(1)} \succeq \dots \succeq X_t^{(k)}$.	
	If the local monotone couplings are used to update each vertex $v \in V$ in $X_t^{(1)},\dots,X_t^{(k)}$, sequentially in the order specified by $\Ord$, 
	then $X_{t+1}^{(1)} \succeq \dots \succeq X_{t+1}^{(k)}$.
	
    We bound the coupling time of the monotone coupling for two instances $\{X_t\}$ and $\{Y_t\}$ of $\GD(\Ord)$.
	Since in monotone systems there are unique maximal and minimal configurations in the partial order, it is sufficient to analyze the coupling time starting from these extremal configurations. Thus,
	suppose that $X_0$ and $Y_0$ are the maximal and minimal configurations, respectively, and let
	$\Tcoup$
	be the coupling time of the monotone coupling starting from these two configurations.
	
	We show that $\Tcoup \le T=c \log n (\log \log n)^2$ for a suitable constant $c>0$. 
	This implies that the mixing time of $\GD(\Ord)$ is $O(\log n (\log \log n)^2)$, as claimed.
	The proof is inductive. For the base case of the induction, observe that if $|V| \le n_0$, where $n_0 \ge 0$ is a large constant we choose later, then  
	we can choose $c = c(n_0)$ large enough such that for any boundary condition on $\boundary V$
	the coupling time bound holds. 
	This is a consequence of the irreducibility of $\GD(\Ord)$ which follows from the assumption that the Glauber dynamics is irreducible; see Section \ref{subsection:prelim:markov-chains}.
	
	Let us assume now inductively that
	for all $d$-dimensional cubes $V' \subset \Z^d$ such that $|V'| \le (4a^{-1}\log n)^d$ (where $a$ is the constant in the definition of SSM),
	any boundary condition on $\boundary V'$ 
	and any ordering $\Ord'$ of the vertices of $V'$ 
	we have that the coupling time of the monotone coupling
	in the subgraph induced by $V'$
	(w.r.t.\ ordering $\Ord'$)
	is at most $c \log |V'| (\log \log |V'|)^2$.
	
	We show that, for all $v \in V$, after $T = c \log n (\log \log n)^2$ steps of the monotone coupling, we have 
	\begin{equation}
	\label{eq:local-bound}
	\Pr[X_T (v) \neq Y_T(v)] \le \frac{1}{4n}.
	\end{equation}
	A union bound over the vertices implies that $\Tcoup \le T$.
	We introduce some notation first.
	
	For $v \in V$ and $\ell > 0$, 
	let $B_v(\ell) \subset V$ 	
	be the intersection of $V$ with the $d$-dimensional cube of $\Z^d$ of side length $2\ell+1$ centered at $v$.
	Let $r = 2a^{-1}\log n$, $B_v = B_v(r)$ and $B^c_v =  V \setminus B_v$.
	
	For each $v\in V$ we consider
	four additional copies of $\GD(\Ord)$: $\{W_t\}$, $\{W_t^\mu\}$, $\{Z_t\}$ and $\{Z_t^\mu\}$.
	These four chains ignore all the updates outside of $B_v$ and their steps are coupled with those of $\{X_t\}$ and $\{Y_t\}$.
	More precisely, for each $u \in V$ (in the order specified by $\Ord$),
	if $u \in B_v$ then the local monotone coupling is used to update the configurations in  $W_t(u)$, $W_t^\mu(u)$, $Z_t(u)$, $Z_t^\mu(u)$, $X_t(u)$ and $Y_t(u)$. Otherwise, if $u \not\in B_v$, the local monotone coupling is used only to update $X_t(u)$ and $Y_t(u)$ and $W_t(u)$, $W_t^\mu(u)$, $Z_t(u)$ and $Z_t^\mu(u)$ are not updated.
	
	We specify next the initial configuration of these chains. 
	We set
	$W_0 = X_0$, 
	$Z_0 = Y_0$,
	$W_0^\mu(B_v^c) = X_0(B_v^c)$ and
	$Z_0^\mu(B_v^c) = Y_0(B_v^c)$.
	To define the configurations of $W_0^\mu$ and $Z_0^\mu$ in $B_v$,
	let $\phi_{\textsc{w}}$ and $\phi_{\textsc{z}}$ be the stationary measures of $\{W_t\}$ and $\{Z_t\}$, respectively.
	These are the distributions induced in $B_v$
	by the configurations in $W_0^\mu(B_v^c)$ and $Z_0^\mu(B_v^c)$, respectively, and possibly the boundary condition $\psi$ on $\boundary V$.
	The configurations in $W_0^\mu(B_v)$ and $Z_0^\mu(B_v)$ are sampled independently from $\phi_{\textsc{w}}$ and $\phi_{\textsc{z}}$, respectively.
	
	Our choice of initial configurations
	and the monotonicity of the coupling imply that
	$W_t \succeq X_t \succeq Y_t \succeq Z_t$
	for all $t \ge 0$. Hence, 
	\begin{align}
	\Pr[X_T(v) \neq Y_T(v)] 
	&\le 
	\Pr[W_T(v) \neq Z_T(v)] \notag\\
	&\le \Pr[W_T(v) \neq W_T^\mu(v)]
	+\Pr[W_T^\mu(v) \neq Z_T^\mu(v)]
	+\Pr[Z_T^\mu(v) \neq Z_T(v)] \label{eq:union-monotone},
	\end{align}
	where the second inequality follows from a union bound. We bound the first and third terms in the right-hand side of (\ref{eq:union-monotone}) using the inductive hypothesis.
	The bound for the inner term follows from SSM.
	
	The chains $\{W_t\}$, $\{Z_t\}$, $\{W_t^\mu\}$ and $\{Z_t^\mu\}$
	are systematic scan dynamics on $B_v$ w.r.t.\ the ordering $\Ord(B_v)$ that $\Ord$ induces on the vertices of $B_v$.
	Since $|B_v| \le (2r)^d = (4a^{-1}(\log n))^d$, for $n_0$ sufficiently large and $n \ge n_0$ 
	\[c \log n (\log \log n)^2 \ge c \log |B_v| (\log \log |B_v|)^2  \log_4 (12 n).\]
	So, the inductive hypothesis and (\ref{eq:prelim:coup-boost}) imply that $\Pr[W_T(v) \neq W_T^\mu(v)] \le 1/(12 n)$. The same bound for $\Pr[Z_T^\mu(v) \neq Z_T(v)]$ can be deduced analogously.
	
	To bound the probability that $W_T^\mu(v) \neq Z_T^\mu(v)$, i.e., the inner term of (\ref{eq:union-monotone}),
	let us assume without of generality that the linear ordering
	on the spins is $q \succeq q-1 \succeq \dots \succeq 1$.
	Since,
	$W_t^\mu(v) \succeq Z_t^\mu(v)$ for all $t \ge 0$, then
	$W_t^\mu(v) \ge Z_t^\mu(v)$.	
	Moreover, the configurations in $W_t^\mu(B_v)$ and $Z_t^\mu(B_v)$ are distribtued according to $\phi_{\textsc{w}}$ and $\phi_{\textsc{z}}$, respectively, for all $t \ge 0$.
	Therefore,
	\begin{align}
	\Pr[W_T^\mu(v) \neq Z_T^\mu(v)] 
	&\le \Exp[W_T^\mu(v) - Z_T^\mu(v)] 
	\le (q-1) {\|\phi_{\textsc{w},v}-\phi_{\textsc{z},v}\|}_{\textsc{tv}} \notag,
	\end{align}
	where $\phi_{\textsc{w},v}$ and $\phi_{\textsc{z},v}$ are
	the distributions induced in $\{v\}$ by $\phi_{\textsc{w}}$ and $\phi_{\textsc{z}}$, respectively. Hence, SSM and a union bound over the boundary of $B_v$ imply that 
	\[\Pr[W_T^\mu(v) \neq Z_T^\mu(v)] \le (q-1)b |\boundary B_v| \exp(-a r) 
	\le \frac{2(q-1)bd(4a^{-1}\log n)^{d-1}}{n^2} 
	\le \frac{1}{12 n},\]
	where in the second inequality we used that $|\boundary B_v| \le 2d(2r)^{d-1}$ and $r=2a^{-1}\log n$ and the last one
	holds for all $n \ge n_0$ and $n_0$ large enough.	
	Putting all these bounds together we get (\ref{eq:local-bound}).
	A union bound over the vertices implies that $\Pr[X_T \neq Y_T] \le 1/4$. Consequently, $\Tcoup \le T = c \log n (\log \log n)^2$ and the mixing time of $\GD(\Ord)$ is $O(\log n (\log \log n)^2)$.
\end{proof}

\begin{proof}[Proof of Theorem \ref{thm:ss:monotone:no-propagation}]
	
	Let $\{X_t\}$, $\{Y_t\}$ be two copies of $\GD(\Ord)$ such that $X_0$ and $Y_0$ are the unique maximal and minimal configurations of the partial order, respectively. 
	We couple these two realizations of $\GD(\Ord)$ with the monotone coupling
	described at the beginning of the proof of Theorem \ref{thm:ss:monotone},
	where we established that the coupling time of this monotone coupling 
	in a $d$-dimensional cube $V$ with an arbitrary boundary condition $\psi$
	is at most $c \log |V| (\log \log |V|)^2$.
	
	Let 
	\[\rho(t) = \max_{v \in V} \Pr[X_t(v) \neq Y_t(v)].\]
	We show that $\rho(T) \le 1/n^2$ for some $T=O(\log n)$. A union bound over the vertices then implies that $\Pr[X_T \neq Y_T] \le 1/n$, and thus 
	$\taumix(\GD(\Ord),1/n) = O(\log n)$.
	Consequently, 
	the
	mixing time of $\GD(\Ord)$ is at most $T=O(\log n)$
	and its
	relaxation time  is $O(1)$
    by (\ref{prelim:eq:gap-lower}).
	
	To bound $\rho$ we establish a recurrence relation.
	Prior to this,
	we show that after $t_0 = \lceil (\log n_0)^2 \log_4 n_0 \rceil$ steps $\rho(t_0) \le 1/n_0$, where $n_0$ is sufficiently large constant. 
	This will provide a stopping point for our recurrence for $\rho$.
	
	As before,
	for $v \in V$ and $\ell > 0$, 
	let $B_v(\ell) \subset V$ 	
	be the intersection of the $d$-dimensional cube of side length $2\ell+1$ centered at $v$ with $V$.
	Let $r = \lfloor n_0^{1/d}/2 \rfloor$ and let $B_v=B_v(r)$. Let $\{W_t\}$ and $\{Z_t\}$ 
	be two auxiliary copies of $\GD(\Ord)$ such that $X_0(B_v) = W_0(B_v)$ and $Y_0(B_v) = Z_0(B_v)$.
	In $B_v^c = V \setminus B_v$,
	$W_0$ and $Z_0$
	have the same fixed configuration; this configuration can be any valid configuration
	provided $W_0(B_v^c) = Z_0(B_v^c)$. 
	
	These four copies of the chain are coupled with the monotone coupling, but $\{W_t\}$ and $\{Z_t\}$ ignore all the updates outside of~$B_v$. 
	That is, for each $u \in V$ (in the order specified by $\Ord$),
	if $u \in B_v$ then the local monotone coupling is used to update the spins of  $W_t(u)$, $Z_t(u)$, $X_t(u)$ and $Y_t(u)$. 
	Otherwise, if $u \not\in B_v$, the local monotone coupling is used only to update $X_t(u)$ and $Y_t(u)$ and $W_t(u)$, $Z_t(u)$ are not updated.
	A union bound implies
	\begin{align}
	\Pr[X_{t_0}(v) \neq Y_{t_0}(v)] ~\le~ 
	\Pr[X_{t_0}(v) \neq W_{t_0}(v)]
	+\Pr[W_{t_0}(v) \neq Z_{t_0}(v)] 
	+\Pr[Z_{t_0}(v) \neq Y_{t_0}(v)]. \notag
	\end{align}
	Let $l = \mathcal{L}(\Ord)$ and
	observe that $r = \lfloor n_0^{1/d}/2 \rfloor >  t_0 l$
	for sufficiently large $n_0$. Thus, $\Pr[X_{t_0}(v) \neq W_{t_0}(v)] = 0$ and $\Pr[Z_{t_0}(v) \neq Y_{t_0}(v)] = 0$, since it is impossible for disagreements to propagate from $\boundary B$ to $v$. 
	Hence, 
	\begin{align}
	\Pr[X_{t_0}(v) \neq Y_{t_0}(v)] ~\le~ 
     \Pr[W_{t_0}(v) \neq Z_{t_0}(v)]. \notag
	\end{align}
	Now,
	let $\Ord(B_v)$ be 
	the ordering induced on $B_v$ by $\Ord$.
	Since 
	$|B_v| \le (2r)^d \le n_0$, the coupling time of the monotone coupling
	for the systematic scan chain on $B_v$ (w.r.t.\ $\Ord(B_v)$)
	is at most $(\log n_0)^2$, provided $n_0$ is sufficiently large (see proof of Theorem \ref{thm:ss:monotone}).
	Hence, since $t_0 = \lceil (\log n_0)^2 \log_4 n_0 \rceil$, (\ref{eq:prelim:coup-boost}) implies that
	$\Pr[W_{t_0}(v) \neq Z_{t_0}(v)] \le 1/n_0$
	and so 
	\begin{equation}
	\label{eq:sp}
	\rho(t_0) \le \frac{1}{n_0}.
	\end{equation}
	
	We establish next our recurrence for $\rho$. We prove that 
	\begin{equation}
	\label{eq:ss:recurrence}
	\rho(2t) \le (4tl)^d \rho(t)^2
	\end{equation}
	for all $t = o((\log n)^2)$. Let $\mathcal{A}$ be the event that $X_t(B_v(2tl)) \neq Y_t(B_v(2tl))$. (The restriction that $t = o((\log n)^2)$ is to ensure that $2tl \ll n$ and avoid unnecessary complications.) Then,
	\[\Pr[X_{2t}(v) \neq Y_{2t}(v)] \le \Pr[X_{2t}(v) \neq Y_{2t}(v) | \mathcal{A}]\Pr[\mathcal{A}] + \Pr[X_{2t}(v) \neq Y_{2t}(v) | \neg \mathcal{A}].\]
	Observe that $\Pr[X_{2t}(v) \neq Y_{2t}(v) | \mathcal{A}] \le \rho(t)$,
	since $\rho(t)$ is the maximum probability of disagreement at any vertex
	assuming the worst possible  pair of staring configurations. Moreover,
	$\Pr[\mathcal{A}] \le |B_v(2tl)| \rho(t)$ by a union bound and $\Pr[X_{2t}(v) \neq Y_{2t}(v) | \neg \mathcal{A}] = 0$ since disagreements can only propagate a distance of at most $tl$ in $t$ steps. 
	Hence, for all $v \in V$,
	\[\Pr[X_{2t}(v) \neq Y_{2t}(v)] \le (4tl)^d \rho(t)^2,\]
	and (\ref{eq:ss:recurrence}) follows.
	
	Finally, we use this recurrence together with the stopping point in (\ref{eq:sp}) to show that $\rho(T) \le 1/n^2$ for some $T = O(\log n)$.
	Let $\phi(t) = (8tl)^d \rho(t)$. Then,
	$\phi(2t) \le \phi(t)^2$, and so for $T = 2^\alpha t_0$ we get
	\[\rho(T) \le \phi(T) \le \phi(t_0)^{T/t_0}.\]
	Since $\phi(t_0) = (8t_0l)^d \rho(t_0) \le (8t_0l)^d/n_0$ and $t_0 = \lceil (\log n_0)^2 \log_4 n_0 \rceil$, for large enough $n_0$ we have $\phi(t_0) \le 1/{\e}$ and thus $\rho(T) \le {\e}^{-T/t_0}$. Taking $T=O(\log n)$ (i.e., $\alpha = O(\log \log n)$), we get $\rho(T) \le 1/n^2$ as desired.
\end{proof}

\bibliographystyle{plain}
\bibliography{ssm_sw}

\pagebreak
\appendix

\end{document}